\documentclass[12pt]{article}
\usepackage{amsmath}
\usepackage{graphicx}
\usepackage{enumerate}
\usepackage{natbib}
\usepackage{url} 
\usepackage[usenames,dvipsnames]{color}

\usepackage{booktabs}
\usepackage{amssymb,amsthm}
\usepackage{algorithm}
\usepackage[flushleft]{threeparttable}
\usepackage{enumerate}
\usepackage[bookmarksopen=false]{hyperref}
\usepackage{setspace}

\newtheorem{algo}{Algorithm}

\newtheorem{assumption}{Assumption}
\newtheorem{theorem}{Theorem}

\begin{document}

	\def\spacingset#1{\renewcommand{\baselinestretch}%
		{#1}\small\normalsize} \spacingset{1}

	
\title{\bf Mahalanobis balancing: a multivariate perspective on approximate covariate balancing}
  \author{Yimin Dai \hspace{.2cm}\\
    School of Mathematics, Sun Yat-sen University, Guangzhou, China, 510275\\
    and \\
    Ying Yan \\
    School of Mathematics, Sun Yat-sen University, Guangzhou, China, 510275\\
	}
\date{}
  \maketitle

	\bigskip
	\begin{abstract}
		In the past decade,  various  exact  balancing-based weighting methods were  introduced to the causal inference literature.  It  eliminates   covariate imbalance by imposing  balancing constraints in a certain  optimization problem, which can  nevertheless be infeasible when there is bad overlap between the covariate distributions in the treated and control groups or when the covariates are high-dimensional.  Recently, approximate balancing was proposed as an alternative balancing framework. It resolves the feasibility issue by using inequality moment constraints instead. However, it can be difficult to select the threshold parameters. Moreover, moment constraints may not fully capture the discrepancy of covariate distributions. In this paper, we propose Mahalanobis balancing to approximately balance covariate distributions from a multivariate perspective. We use a  quadratic constraint  to control  overall imbalance with a single threshold parameter, which can be tuned by  a simple  selection procedure. We show that the dual problem of Mahalanobis balancing is an $\ell_2$ norm-based regularized regression problem, and  establish interesting  connection  to  propensity score models. We derive asymptotic properties, discuss the high-dimensional scenario, and make extensive numerical comparisons with existing balancing methods.
	\end{abstract}
	
	\noindent%
	\vfill
	
	\newpage
	\spacingset{1.5} 
	
	\section{Introduction}
	\label{sec1}

	Inference about causation gains increasing attention in medical science, economics,  computer science, and many other disciplines. Propensity score \citep{rosenbaum1983central}, the probability of receiving a treatment conditional on the  covariates,  plays a central role in causal inference for observational studies.
	There are several classes of propensity score-based methods in practitioners'  toolkit, including matching, weighting, and subclassification. We focus on  weighting in this article \citep{rosenbaum1987model,robins2000marginal,hirano2003efficient}. Inverse probability weighting and its doubly robust version are perhaps  the most commonly used weighting methods. They  explicitly estimate the propensity score via a parametric or nonparametric model \citep[e.g.][]{hirano2003efficient}. However, as demonstrated by \cite{kang2007demystifying} and  related articles, inverse probability weighting is  prone to misspecification of the propensity score model, and covariate balance may not be attained after reweighting.
	
	In the past decade, numerous robust weighting methods  have been proposed, which  aim to    directly balance covariates in the estimation procedure.  \cite{hainmueller2012entropy} introduced the entropy balancing method, which optimizes the entropy loss function subject to a set of  balancing constraints. The balancing constraints enforce  the weighted moments in the control group to be equal to the unweighted counterparts in the treated group, and thus   finite-sample  {exact balance} is achieved.  The  balanced weights are then used  to reweight the  subjects in the control group to produce a weighted estimator of the average treatment effect on the treated. This method has close connection to    survey sampling, missing data, and machine learning  \citep[e.g.][]{dehejia1999causal,Graham2012,mohri2018foundations}. \cite{zhao2016entropy} revealed that entropy balancing is doubly robust.   Recently, entropy balancing  has received great attention in applied areas. For example, it was  adapted to integrate randomized clinical trial and observational study \citep{lee2023improving}, to transport trial results to a target population \citep{Josey2021b}, and so forth. \cite{imai2014covariate} considered  parametric  propensity score model   and solved the likelihood score function together with the  balancing constraints simultaneously to produce covariate balancing propensity score. \cite{fan2022optimal}  improved the covariate balancing propensity score method. \cite{chan2016globally} constructed a class of calibration weighting estimators by minimizing a  distance measure subject to balancing  constraints.  Empirical likelihood and exponential tilting belong  to this class.  \cite{YiuSu2018} proposed to eliminate  the association between treatment and covariates in the weight construction procedure.  \cite{Hazlett2020} recommended to construct balancing constraints from the kernel viewpoint. \cite{Josey2021} proposed a Bregman distance framework which unifies many existing methods.
	
	A common feature of the aforementioned balancing methods is that  the balancing conditions are directly imposed as {equality moment constraints}  in a  convex optimization problem. We call them  \textit{exact balancing} methods in this paper.  Compared to inverse probability weighting, exact balancing is capable of producing stable weights and is more robust to model misspecification  in many circumstances. Nevertheless, exact balancing performs poorly in the bad-overlap scenario  where
	there is limited overlap between the covariate distributions in the treated and control groups, or in the high-dimensional scenario where the number of constraints is large. Even worse,  the optimization problem may be infeasible in these difficult scenarios. We refer readers to a discussion of convex hull problems in \cite{Hayakawa2023}, which implies  infeasibility of exact balancing in the high-dimensional scenario.
	
	To alleviate the infeasibility problem,  \cite{zubizarreta2015stable} and \cite{wang2020minimal} proposed  an \textit{approximate balancing} framework, i.e., the  minimal dispersion {approximately balancing} weights (MDABW) method, which replaces
	the exact balancing conditions with  less restrictive   {inequality constraints}. The inequality constraints are obtained by  controlling  univariate  dispersion of each covariate. \cite{wang2020minimal} showed that  the  dual  problem of  MDABW  is  a weighted $\ell_1$ norm-based  regularized regression problem. There are several other robust balancing methods which may not suffer the infeasibility problem. For example, \cite{wong2018kernel} directly minimized covariate functional imbalance. \cite{Zhao2019} proposed a general balancing framework using  tailored loss functions.  \cite{LiMorganZaslavsky2018} and \cite{MaWang2020} provided  insights about robust weighting from different perspectives.

	In general, covariate balancing becomes much more difficult  in the high-dimensional scenario, because the  overlap condition tends to be  more restrictive as the dimensionality increases. We refer the readers to \cite{d2021overlap} for a formal discussion of the  overlap condition in the high-dimensional scenario.  Regularization   plays a key role  in high-dimensional analysis, and it has been adapted to  covariate balancing  recently. Among others, \cite{athey2016approximate,ning2020robust,tan2020model,tan2020regularized} addressed  high-dimensional balancing problems using  regularization techniques. We call them    high-dimensional regularized balancing methods. Many  of these  methods involved outcome models and thus cannot be categorized as   data preprocessing procedures. Interestingly,  \cite{athey2016approximate,ning2020robust,tan2020regularized}  established  tight connections  of their methods to the  MDABW  method.

	Now we discuss  potential limitations of  the MDABW method as an approximate balancing framework. The MDABW method controls the  weighted  absolute standardized mean difference (ASMD)  for each covariate with a threshold parameter, where the ASMD is a  univariate  imbalance measure frequently used in observational studies.   {Approximate balance} is achieved for each covariate separately by controlling the  threshold parameters, and  thus   the MDABW method is a \textit{univariate approximate balancing} framework. However,   it can be difficult to select the threshold parameters simultaneously. For example,  an attempt to reduce imbalance of some covariate by shrinking the corresponding threshold parameter may reversely  enlarge imbalance of other covariates. This dilemma is not uncommon in  the bad-overlap  scenario. Second, even  if univariate approximate balancing is achieved  for all  constraints separately,  the  overall imbalance  is not necessarily under control.  Moreover, the MDABW method requires  that the number  of  constraints is much smaller than the sample size. Our numerical studies demonstrate that the MDABW method may perform unsatisfactorily in the bad-overlap and  high-dimensional scenarios.
	
	The limitation of  univariate approximate balancing   motivates the proposed multivariate approximate balancing method in this paper. We show that  univariate approximate balancing  can be greatly improved by monitoring and controlling  overall covariate balance from the multivariate perspective.  The concept of multivariate imbalance measure   is widely adopted  in   the  matching literature to quantify  the discrepancy  of the covariate distributions
	\citep[e.g.][]{Iacus2011,Iacus2012,DiamondSekhon2013,imbens2015causal}. However,   it  is  largely neglected in the covariate balancing literature. In this paper, we highlight  the importance of controlling multivariate imbalance.
	
	The main contribution of this paper is that we propose  Mahalanobis balancing as a \textit{multivariate approximate balancing} framework. We use a generalized version of Mahalanobis distance to measure   multivariate imbalance of the covariates between groups. A sufficiently small value of the  multivariate imbalance measure suggests that the covariates are approximately balanced from the multivariate perspective. The multivariate balancing condition is encoded as  a  single quadratic inequality constraint in a convex optimization problem with a pre-specified loss function (e.g.,  entropy, empirical likelihood).  We study the primal optimization problem  by utilizing the Fenchel duality theory \citep[e.g.][]{Bertsekas2016,mohri2018foundations,MordukhovichNam2022}. We show that the dual problem is an unconstrained regularized regression problem with an   $\ell_2$ norm  regularizer. Off-the-shelf convex optimization algorithms can be employed to solve the dual problem. We adopt  the BFGS   quasi-Newton algorithm in this paper.

	Mahalanobis balancing differs from the MDABW method by that   multivariate imbalance, rather than univariate imbalance, is controlled. This simple idea has important consequences that make Mahalanobis balancing a highly competitive balancing method.  First,   a univariate  threshold parameter is attached to the quadratic inequality constraint in Mahalanobis balancing, and thus the  issue of tuning  multiple or high-dimensional threshold parameters is resolved. Second, there are situations where univariate approximate balancing is achieved but the covariate distributions remain   imbalanced after reweighting (e.g., Scenarios B, C, E in Section \ref{sec4}). The MDABW method does not perform satisfactorily in these situations, and   treatment effect estimation can be considerably biased. It suggests that  controlling univariate imbalance  in reweighting can be insufficient to delineate multivariate imbalance and obtain valid treatment effect estimation.   In comparison,   Mahalanobis balancing   greatly improves the  performance of univariate approximate balancing in these situations, which confirms the importance of  controlling  multivariate imbalance. Third,  Mahalanobis balancing is not restricted to the low-dimensional situation.  In the Supplementary Material, we propose a high-dimensional version of  Mahalanobis balancing where the idea of  multivariate imbalance is further exploited. The simulations  reveal that it  is competitive to  the   high-dimensional regularized balancing methods \citep{athey2016approximate,ning2020robust,tan2020regularized}. In comparison, the MDABW method is restricted to the low-dimensional scenario.

	The rest of the paper is organized as follows. In Section \ref{sec2}, we introduce necessary notations,  existing imbalance measures, and balancing methods. In Section \ref{sec3}, we present the proposed Mahalanobis balancing method, and discuss the high-dimensional situation. Asymptotic properties of   Mahalanobis balancing  are studied. In Section \ref{sec4}, we make extensive comparison of   Mahalanobis balancing  to  the existing  balancing methods in numerical studies. We draw conclusion in the last section. Extended discussion of the high-dimensional setting,  proofs, and additional numerical results are available in the online Supplementary Material. An R package of  Mahalanobis balancing and numerical examples  are available  at Github  via the  link:
\verb"https://github.com/yimindai0521/ACBalancing/".

	\section{Notation and Preliminary}\label{sec2}
	\subsection{The framework}\label{sec2.1}
	Consider a simple random sample of  $n$ subjects from a population. Let $T$ be a binary treatment indicator with $T=1$ if the subject is treated and $T=0$ otherwise. Let $X=(X_{1}, \cdots, X_{p})^\top$ be a $p$-dimensional  vector of pre-treatment covariates. Let $Y$ be the univariate outcome.  The observed data are $\left\{\left({X}_{i}, T_{i}, Y_{i}\right): i=1, \cdots, n\right\}$, which are $n$ independent and identically distributed copies of the triplet $(X, T, Y)$. Under the potential outcome framework \citep[e.g.][]{imbens2015causal}, we define a pair of potential outcomes  $\left\{Y(0), Y(1)\right\}$ for each subject if he were not treated or treated. The observed outcome  is $Y=\left(1-T\right) Y(0)+T Y(1).$
	
	We impose the following strong ignorability and overlap assumptions for identification and inference of the average treatment effect.
	
	\begin{assumption}[Strong Ignorability]
		$\{Y(0), Y(1)\} \perp T \mid {X}$.
	\end{assumption}
	
	Here,  $\perp$ denotes independence. Assumption 1 states that the potential outcomes $\{Y(0), Y(1)\}$ are independent of the treatment indicator $T$ given pre-treatment covariates ${X}$, which implies that there are no unmeasured confounders. The probability   $\pi({X})=\text{Pr}(T=1 \mid {X})$  is the  propensity score function. Assumption 1   implies $\pi({X})$ is a balancing score \citep{rosenbaum1983central} in the sense that
	$\{Y(0), Y(1)\} \perp T \mid \pi({X})$.  Under Assumption 1, the average treatment effect is identifiable.
	
	Moreover, to ensure that there are informative observations for treatment effect estimation, we impose the overlap assumption:
	
	\begin{assumption}[Overlap]
		$0<\text{Pr}(T=1 \mid {X})<1\text { for any } {X}$ in its support.
	\end{assumption}

	In this paper, the causal estimand of interest is  the average treatment effect (ATE), defined by $\tau=E\{Y(1)-Y(0)\}=\mu_{1}-\mu_{0}$, where $\mu_{j}=E\{Y(j)\}$, $j=0,1$.  The proposed method  can be  adapted to other causal estimands, e.g., average treatment effect on the treated (ATT).   Under Assumption 1, $\tau=E\{\mu_{1}({X})-\mu_{0}({X})\}$, where  $\mu_{j}({X}) = E\{Y(j) \mid X\}$,  $j=0,1$.
	
	A generic reweighting scheme constructs weights ${w}=\left(w_{1}, \cdots, w_{n}\right)^\top$ by exploiting the sampled data and possibly external information from the population, and then estimates the ATE by  $\hat{\tau}=\sum_{i=1}^{n} T_{i} w_{i} Y_{i}-\sum_{i=1}^{n}\left(1-T_{i}\right) w_{i} Y_{i}$.  The  covariate balancing methods introduced in Section \ref{sec1}  explicitly balance the covariates  in the reweighting procedure and produce a set of balanced weights.
	
	\subsection{Imbalance measures}\label{sec2.2}
	A set of pre-specified basis functions, denoted by $\phi_1(X),\cdots, \phi_K(X)$,  can be used instead  of the original covariates $X$ in the construction of balanced weights \citep[e.g.,][]{hainmueller2012entropy,imai2014covariate,wang2020minimal}. In Section \ref{sec3}, we discuss the choice of the basis functions.  Define $\Phi(X)=(\phi_1(X),\cdots, \phi_K(X))^\top$ to be a vector of  basis functions, which is an $\mathbb{R}^p\rightarrow \mathbb{R}^K$ feature mapping.
	
	Assessment of  covariate balance is crucial in the inference of causal effects. Many imbalance measures  are proposed and widely used in the literature.
	The absolute standardized mean difference (ASMD) is the most popular univariate imbalance measure \citep[e.g.,][]{imbens2015causal}. The ASMD for the $k$th basis function $\phi_k(X)$ is expressed as
	\begin{equation}
		\text{ASMD}_k=\frac{\left|\bar{\phi}_{1,k}-\bar{\phi}_{0,k}\right|}
		{\sqrt{\left(s_{1,k}^{2}+s_{0,k}^{2}\right) / 2}}, \ \ k=1,\cdots, K,
		\label{3}
	\end{equation}
	where $\bar{\phi}_{j,k}=\sum_{i=1}^{n}I(T_i=j)  \phi_k(X_{i})/n_j$  and $s_{j,k}^{2}=\sum_{i=1}^{n}I(T_i=j) (\phi_k(X_{i})-\bar{\phi}_{j,k})^2/(n_{j}-1)$ are the treatment-specific sample mean and sample variance for  $\phi_k$, $j=0,1$. Here, $I(\cdot)$ is the indicator function, and  $n_1$ and $n_0$ are the sample sizes of the treated and control groups, respectively.
	When the ASMD for a  basis function does not exceed a pre-specified threshold parameter $\delta$, then this basis function is considered to be approximately balanced. Common choices of $\delta$ are 0.1, 0.2 and  0.25 \citep[e.g.][]{Stuart2013,imbens2015causal,Xie2019}. However,  even when all basis functions are  approximately balanced in the sense that the average of $\text{ASMD}_1,\cdots, \text{ASMD}_K$ is smaller than the threshold, say $\delta=0.2$,  the difference in the outcomes as an unadjusted average treatment effect estimation  can be  still  severely biased. This can be explained by that the ASMD is  a univariate imbalance measure which does not fully characterize    multivariate imbalance.
	
	The targeted absolute standardized mean difference (TASMD), introduced by  \cite{chattopadhyay2020balancing}, is defined as
    $$
    \operatorname{TASMD}_{j,k}=\frac{\left|\bar{\phi}_{j,k}-\bar{\phi}_{k}\right|}{s_{k}},
    $$
    where  $\bar{\phi}_{k} = \sum_{i=1}^{n}\phi_k(X_{i})/n$ and $s_{k}^{2}=\sum_{i=1}^{n}(\phi_k(X_{i})-\bar{\phi}_{k})^2/(n-1)$, $j =0,1; k=1,\cdots, K$.  The ASMD  quantifies the disparity between the treated and control groups, whereas  the TASMD measure quantifies the disparity of each group in relation to a target population. In practical applications, we can employ the treatment-adjusted subgroup mean difference  method to evaluate balance in relation to various target populations. Examples include the population of treated units when estimating the ATT, the entire population when estimating the ATE, or a specific population defined by its observed covariates when estimating the conditional average treatment effect (CATE).

	Multivariate  imbalance measure can  be used to assess the imbalance for all covariates or basis functions simultaneously.  The squared Mahalanobis distance (MD) \citep{imbens2015causal}  is a popular multivariate imbalance measure, given by
	\begin{equation*}
		\text{MD}=(\bar{\Phi}_{1}-\bar{\Phi}_{0})^{\top}\left( \frac{\hat{\Sigma}_1+\hat{\Sigma}_0}{2}\right)^{-1}(\bar{\Phi}_{1}-\bar{\Phi}_{0}),
		\label{4}
	\end{equation*}
	where $\bar{\Phi}_{j}=(\bar{\phi}_{j,1},\cdots, \bar{\phi}_{j,K})^\top$, $j=0,1$, and $\hat{\Sigma}_1$ and $\hat{\Sigma}_0$ are the sample covariance matrices  of $\Phi(X)$ in the treated  and control groups, respectively.  The MD was used in Mahalanobis distance matching \citep[e.g.][]{imbens2015causal} to determine the closeness of subjects between the treated and control groups. Moreover, it was adapted to  genetic matching  \citep{DiamondSekhon2013}, which aims at minimizing a generalized version of MD to optimize postmatching covariate balance. The $\ell_1$ distance-based measure proposed by \cite{Iacus2011} is another example of multivariate imbalance measure, which motivates a  class of matching methods with the monotonic imbalance bounding property. \cite{Iacus2012}  derived  coarsened exact matching   from this class, and emphasized the importance of optimizing multivariate balance in matching. \cite{zhu2018kernel} proposed to use the Kernel distance as a multivariate imbalance measure. \cite{HulinMak2020} studied  energy distance-based covariate balancing.

	\subsection{Exact balancing v.s. approximate balancing}\label{sec2.3}
	
	The existing covariate balancing methods  achieve  exact balance by imposing a set of equality balancing constraints
	\begin{equation}
		\begin{aligned}\label{equalityconstraints}
			\sum_{i=1}^{n} w_{i} T_{i}  \phi_k(X_{i})&= \bar{\phi}_k,\\
			\text{and}\  \ \sum_{i=1}^{n} w_{i} (1-T_{i})  \phi_k(X_{i})&= \bar{\phi}_k,\ \ k=1,\cdots, K,
		\end{aligned}
	\end{equation}
	in the  construction of the weights $w_1,\cdots, w_n$, where $\bar{\phi}_k=\sum_{i=1}^{n}  \phi_k(X_{i})/n$. These constraints  force  the first moment of
	the basis functions to be exactly balanced after reweighting. Examples of  exact balancing methods include entropy balancing \citep{hainmueller2012entropy}, covariate balancing propensity score \citep{imai2014covariate}, and calibration weighting \citep{chan2016globally}, among others.

	Define a weighted version of the ASMD by
	\begin{equation*}
		\text{ASMD}_{k}^w=\frac{\left|\sum_{i=1}^{n} w_{i} T_{i}  \phi_k(X_{i})-\sum_{i=1}^{n} w_{i} (1-T_{i})  \phi_k(X_{i})\right|}{\sqrt{\left(s_{1,k}^{2}+s_{0,k}^{2}\right) / 2}}. \ \ k=1,\cdots, K\label{wASMD}
	\end{equation*}
	The $\text{ASMD}_{k}^w$ is a weighted univariate imbalance measure that quantifies  remaining imbalance for the $k$th basis function $\phi_k$ after reweighting. It  is used in this paper to compare the univariate balancing performance of  weighting methods. Usually, we normalize the weights for each group such that $\sum_{i=1}^n T_i w_i=1$ and  $\sum_{i=1}^n (1-T_i) w_i=1$ before comparison.
	It is straightforward that the equality balancing constraints \eqref{equalityconstraints} lead to that $\text{ASMD}_{k}^w=0$ for all $k=1,\cdots, K$. Therefore,  finite-sample univariate exact balance is achieved  by the exact balancing methods. In comparison, inverse probability weighting does not possess this  attractive property.

	When exact balancing is not attainable in the bad-overlap situation, the MDABW method \citep{zubizarreta2015stable,wang2020minimal}   can be used instead.  The balanced weights for the treated group $\{i: T_i=1\}$ are the global minimum of the following optimization problem:
	\begin{equation}
		\begin{array}{ll}
			\underset{{w}}{\operatorname{minimize}} & \sum_{i=1}^{n} {T}_{i} f\left(w_{i}\right) \\
			\text { subject to } & \left|\sum_{i=1}^{n} w_{i} T_{i} \phi_{k}(X_{i})- \bar{\phi}_k\right| \leq \delta_{k}, k=1, \cdots, K,
		\end{array}
		\label{6}
	\end{equation}
	where $f(\cdot)$ is a pre-specified loss function, and $\delta_1,\cdots, \delta_K\geq 0$ are the  threshold parameters. Because the constraints are imposed separately for the basis functions, MDABW is a univariate approximate balancing method.
	The choices of $f(\cdot)$ is discussed in Section \ref{sec3}. Similarly, the weights for the control group $\{i: T_i=0\}$ can be obtained by replacing $T_i$ with $1-T_i$ in the   optimization problem \eqref{6}.
	
	If we normalize the basis function $\phi_{k}(X_{i})$ to be $\phi_{k}(X_{i})/\sqrt{\left(s_{1,k}^{2}+s_{0,k}^{2}\right) / 2}$  in the   optimization problem \eqref{6}, then
	$\text{ASMD}_{k}^w\leq 2\delta_k$  for $k=1,\cdots, K$.
	A small value of $\delta_k$, say $0.1$, indicates that the $\phi_{k}$ is approximately balanced. However, it is difficult to define optimality for all  threshold parameters $(\delta_1,\cdots,\delta_K)$ simultaneously   in the bad-overlap  setting, because decreasing univariate imbalance for some  basis function may inevitably increase  imbalance for other basis functions. Correlation of the basis functions are not taken into account in Problem \eqref{6}. Tuning a large number of parameters is very time-consuming, and  there is lack of guideline to tune these  parameters simultaneously. These potential limitations  motivate   Mahalanobis balancing in the following section.

	We define the following Mahalanobis imbalance measure (MIM) as a weighted version of the squared Mahalanobis distance:
	\begin{equation}
		\begin{aligned}\label{wMD}
			\text{MIM}^w=&\left\{\sum_{i=1}^{n} w_{i} T_{i} \Phi(X_i)-\sum_{i=1}^{n} w_{i} (1-T_{i}) \Phi(X_i)\right\}^{\top}\left( \frac{\hat{\Sigma}_1+\hat{\Sigma}_0}{2}\right)^{-1}\\
			&\times \left\{\sum_{i=1}^{n} w_{i} T_{i} \Phi(X_i)-\sum_{i=1}^{n} w_{i} (1-T_{i}) \Phi(X_i)\right\},
		\end{aligned}
	\end{equation}
	which can be used to measure the  remaining overall imbalance after reweighting. Note that  $\text{ASMD}_{k}^w=0$ for all $k=1,\cdots, K$ is equivalent to $\text{MIM}^w=0$. However, univariate approximate balance  does not imply   multivariate approximate balance, because it is possible that the $\text{ASMD}_{k}^w$ is small for all $k=1,\cdots, K$ but the $\text{MIM}^w$ is  large.

	\section{Proposed Methodology}\label{sec3}
	\subsection{Mahalanobis balancing}\label{sec3.1}
	
	To alleviate the  limitation of univariate  approximate balancing, we propose the Mahalanobis  balancing (MB) method as a multivariate   approximate balancing framework, which directly controls multivariate imbalance to produce balanced weights.  Specifically, Mahalanobis balancing obtains the balanced weights   for the treated group by solving the  convex optimization problem:
	\begin{equation}
		\begin{array}{ll}
			\underset{w \in \mathbb{R}^n}{\text{minimize}}:& \ \ \sum_{i=1}^{n}T_if(w_i)\\
			\text{subject to:}&\left\{
			\begin{array}{ll}
				w_i\geq 0,\ \  i\in\{j: T_j=1\}; \\
				\sum_{i=1}^{n}T_iw_i\{\Phi(X_i)-\bar{\Phi}\}^{\top}{W}\sum_{i=1}^{n}T_iw_i\{\Phi(X_i)-\bar{\Phi}\} \leq \delta,
			\end{array}
			\right.
		\end{array}
		\label{9}
	\end{equation}
	where    $f(\cdot)$ is a convex loss function, $\Phi(X)$ is a vector of basis functions defined in Section \ref{sec2.2}, $\bar{\Phi}=\sum_{i=1}^{n}\Phi(X_i)/n$, ${W}$ is a user-specified $K\times K$ positive-definite weight matrix, and $\delta\geq 0$ is a  univariate threshold parameter. The  MB  weights are  obtained after normalization for the treated subjects, i.e., $w_i^{MB}=w_i/\sum_{j=1}^n T_j w_j$. Similarly, the  MB weights for the control group  can be obtained.
	
	The role of each $w_i$ is to down-weigh or up-weigh the deviation of  $\Phi(X_i)$ from the pooled-sample mean $\bar{\Phi}$. The  matrix ${W}$ weighs the relative importance of the basis functions $\phi_1(\cdot),\cdots, \phi_K(\cdot)$.    In this paper, we consider  two choices of ${W}$: (i) the diagonal matrix ${W}_1=[\text{diag}\{ (\hat{\Sigma}_1+\hat{\Sigma}_0)/2\}]^{-1}$, where its main diagonal contains    the reciprocals of the diagonal elements of $(\hat{\Sigma}_1+\hat{\Sigma}_0)/2$; (ii) ${W}_2=\{ (\hat{\Sigma}_1+\hat{\Sigma}_0)/2\}^{-1}$. We use ${W}_2$  when the dimension $K$ is not large and  the basis functions are not highly correlated. Otherwise, we use ${W}_1$. We choose ${W}$ in the data-preprocessing step to normalize the data, and thus treat it as fixed in the theoretical development.
	
	For an arbitrary set of  normalized  weights for the treated subjects,  we define the  generalized  Mahalanobis imbalance measure (GMIM) for the treated  group:
	\begin{equation}
		\begin{array}{lll}
			\text{GMIM}_1^w=&\left\{\sum_{i=1}^{n}T_iw_i\Phi(X_i)-\bar{\Phi}\right\}^{\top}{W}\left\{\sum_{i=1}^{n}T_iw_i\Phi(X_i)-\bar{\Phi}\right\}.
		\end{array}
		\label{GMD}
	\end{equation}
	Similarly, we define $\text{GMIM}_0^w$ for the control group.
	The $\text{GMIM}_1^w$ measures the residual multivariate difference between the basis functions in the treated group and the sample average $\bar{\Phi}$ after reweighting. Note that  the  term $\sum_{i=1}^{n}T_iw_i\{\Phi(X_i)-\bar{\Phi}\}^{\top}{W}\sum_{i=1}^{n}T_iw_i\{\Phi(X_i)-\bar{\Phi}\}$  in Problem \eqref{9} is the unnormalized version of the $\text{GMIM}_1^w$. Therefore,  Problem \eqref{9} explicitly
	restricts the unnormalized residual multivariate imbalance by a single threshold parameter $\delta$.  When  Problem \eqref{9} is feasible for $\delta=0$, then MB reduces to  exact balancing.   When Problem \eqref{9} is infeasible for  $\delta=0$, we need to tune a positive value for $\delta$. In Section \ref{sec3.2}, we propose  to optimize post-weighting multivariate balance  by monitoring  $\text{GMIM}_1^w$ in the  selection procedure.
	
	We remark that we prefer to  minimize  $\text{GMIM}_1^w$ and  $\text{GMIM}_0^w$ separately rather than minimize the  Mahalanobis imbalance measure  $\text{MIM}^w$ in equation \eqref{wMD}. First, it allows us to easily obtain balanced weights from two  separate optimization problems. Second,  by minimizing $\text{GMIM}_1^w$ and  $\text{GMIM}_0^w$, the  Mahalanobis imbalance measure  $\text{MIM}^w$ is under control.  More importantly,   compared to the $\text{MIM}^w$, the $\text{GMIM}_1^w$ and $\text{GMIM}_0^w$ are more relevant to  the post-weighting  multivariate balancing performance. In particular,  even if the $\text{MIM}^w$ is very small, it does not imply that the   weighted basis functions $\sum_{i=1}^{n}T_iw_i\Phi(X_i)$ and $\sum_{i=1}^{n}(1-T_i)w_i\Phi(X_i)$ are close to the sample average  $\bar{\Phi}$. That is, it is possible that the weighted distributions in the two groups are close to each other, but meanwhile they are  different from the  underlying distribution in the target population so that the $\text{GMIM}_1^w$ and $\text{GMIM}_0^w$ are  large.  This phenomenon frequently occurs in the bad-overlap and high-dimensional settings.

	Next, we discuss the choice of the loss function $f(\cdot)$. One may use  the entropy  function $f(x)=x\log(x)$ \citep{hainmueller2012entropy}, the negative of empirical likelihood $f(x)=-\log(x)$, the quadratic function   by  \cite{zubizarreta2015stable}, the distance measure by \cite{chan2016globally}, and the Bregman distance by \cite{Josey2021}, among others. We prefer  the entropy  function $f(x)=x\log(x)$ for its stable performance and theoretical properties.
	
	The choice of the basis functions is important for the consistency  of the balancing methods \citep[e.g.][]{zhao2016entropy}.
	In practise, we consider  the first and second moments of $X$ as the   basis functions, and sometimes the interaction terms are included. Motivated by the  theory of Reproducing Kernel Hilbert Space (RKHS)  \citep{wong2018kernel,Zhao2019,Hazlett2020}, we  also consider  the second choice $\Phi(X)=(K(X,X_1),\cdots,K(X,X_n))^\top$ when the dimension of the covariates $p$ is large compared to the sample size $n$. Here, $K(\cdot,\cdot)$ is a pre-specified kernel function \citep[e.g.,][]{mohri2018foundations}, and the number of basis functions is $K=n$.

	In the following, we study the dual  of the primal optimization problem  \eqref{9}. By the Fenchel duality theory \citep{Bertsekas2016,mohri2018foundations, MordukhovichNam2022},  we show that the dual problem is a regularized propensity score model   with  an $\ell_2$ norm penalty. In Theorem \ref{thm1}, we formally establish the connection between Mahalanobis balancing and    $\ell_2$ shrinkage estimation of the propensity score model.
	
	We define some additional notations. Apply  Cholesky decomposition to the weight matrix  ${W}$ and write ${W}=({W}^{1/2})^\top{W}^{1/2}$. Let $\lVert\theta\rVert_2=\sqrt{\theta_1^2+\cdots+\theta_K^2}$ be the $\ell_2$ norm for an arbitrary $K$-dimensional vector $\theta=(\theta_1,\cdots,\theta_K)^\top$. The quadratic inequality constraint in  Problem \eqref{9} can then be rewritten as
	$\left.\lVert\sum_{i=1}^{n}T_iw_i{W}^{1/2}(\Phi(X_i)-\bar{\Phi})\right.\rVert_2 \leq \sqrt{\delta}$.
	Let $\tilde{f}(p)=f(p)$ if $p\geq 0$, and $\tilde{f}(p)=+\infty$, otherwise. Let $\partial g(x)$ be the subgradient of the function $g(\cdot)$ at $x$ in the  domain of the function \citep{mohri2018foundations}. Define a convex subset  $\mathcal{C}=\{u\in \mathbb{R}^K:\ \lVert u \rVert_2\leq \sqrt{\delta}\}$. Define $I_{\mathcal{C}}(u)=0$ if $u\in \mathcal{C}$, and  $I_{\mathcal{C}}(u)=+\infty$ otherwise. Then,  Problem \eqref{9} is equivalent to the following unconstrained primal optimization problem:
	\begin{equation}
		\underset{w \in \mathbb{R}^{n_1}}{\text{minimize}}: \ \sum_{i=1}^{n}T_i\tilde{f}(w_i)+I_{\mathcal{C}}\left(\sum_{i=1}^{n}T_iw_i{W}^{1/2}(\Phi(X_i)-\bar{\Phi})\right).
		\label{Primal}
	\end{equation}
	
Let $\tilde{f}^{*}(\cdot)$ be the conjugate function of $\tilde{f}(\cdot)$ \citep{Bertsekas2016}.  The following theorem states the dual problem.
	\begin{theorem} \label{thm1}
		  The dual of Problem  \eqref{Primal} is the following unconstrained optimization problem:
		\begin{equation}
			\underset{\theta \in \mathbb{R}^K}{\text{minimize}}: \ \sum_{i=1}^{n}T_i\tilde{f}^*\left(\theta^\top {W}^{1/2}(\Phi(X_i)-\bar{\Phi})\right)+\sqrt{\delta}\lVert \theta\rVert_2.
			\label{Dual}
		\end{equation}
		Let $\hat{w}=\{\hat{w}_i: T_i=1\}$ be the  solution of the primal problem \eqref{Primal}  and $\hat{\theta}=(\hat{\theta}_1,\cdots,\hat{\theta}_K)^\top$ be the  solution of the dual problem \eqref{Dual}. The $\hat{w}_i$ can be expressed as a function of  $\hat{\theta}$:
		\begin{equation}
			\hat{w}_i=\partial(\tilde{f}^{*})\left(\hat{\theta}^{\top} {W}^{1/2}(\Phi(X_i)-\bar{\Phi})\right), \ \ \text{for}\ \  i\in\{j: T_j=1\}.
			\label{relation}
		\end{equation}
	\end{theorem}
	
In the rest of the paper, we consider the simpler situation that the loss function $f(\cdot)$ is differentiable and that it has non-negative domain as in \cite{Josey2021}. Entropy, KL divergence, empirical likelihood, logistic loss \citep{tan2020regularized} and many other choices satisfy this condition, whereas the quadratic function in \cite{zubizarreta2015stable} does not. This condition allows us to consider gradient instead of subgradient, and obtain $\hat{w}_i=(\tilde{f}^{*})^{\prime}\left(\hat{\theta}^{\top} {W}^{1/2}(\Phi(X_i)-\bar{\Phi})\right)=(\tilde{f}^{\prime})^{-1}(\hat{\theta}^{\top}W^{1 / 2}\left(\Phi\left(X_{i}\right)-\bar{\Phi}\right))$, where  $(\tilde{f}^{*})^{\prime}(\cdot)$ is the  first derivative of  $\tilde{f}^{*}(\cdot)$.	For example, when we use the entropy  $f(x)=x\log(x)$ as the loss function, the dual problem \eqref{Dual} becomes
	$$\underset{\theta \in \mathbb{R}^K}{\text{minimize}}: \ \sum_{i=1}^{n}T_i\exp\left\{\theta^\top {W}^{1/2}(\Phi(X_i)-\bar{\Phi})-1\right\}+\sqrt{\delta}\lVert \theta\rVert_2,$$
	and the estimated balanced weight $\hat{w}_i=\exp\left\{\hat{\theta}^\top {W}^{1/2}(\Phi(X_i)-\bar{\Phi})-1\right\}$.
	The  estimated MB weight, denoted by $\hat{w}_i^{MB}$, is then obtained by normalization:
	\begin{equation}
		\hat{w}_i^{MB}=\frac{\hat{w}_i}{\sum_{j:T_j=1}\hat{w}_j}=\frac{\exp\left\{\hat{\theta}^\top {W}^{1/2}\Phi(X_i)\right\}}{\sum_{j:T_j=1}\exp\left\{\hat{\theta}^\top {W}^{1/2}\Phi(X_j)\right\}}, \ \ \text{for all}\ \ i\in\{j: T_j=1\}.\label{MBweight}
	\end{equation}
	When the propensity score $\pi(x)=\text{Pr}(T=1 \mid {x})$ follows the log model: $\log(\pi(x;\beta))=\beta^\top \Phi(x)$,  the inverse probability weight $1/\{n\pi(X_i;\beta)\}$ coincides with  the expression of the unnormalized weight $\hat{w}_i$ when $\Phi(\cdot)$ includes an intercept. Moreover,  the normalized inverse probability weight
	$\pi^{-1}(X_i;\beta)/\{\sum_{j:T_j=1}\pi^{-1}(X_j;\beta)\}$ coincides with the expression of the MB weight  $\hat{w}_i^{MB}$. There is a one-to-one correspondence of the dual parameter  $\theta$ and the  coefficient $\beta$ in the  propensity score model $\pi(x;\beta)$: $\beta=(W^{1/2})^\top\theta$. Therefore, solving the dual problem \eqref{Dual} is equivalent to fitting an  $\ell_2$ norm-based regularized generalized linear  model with the logarithm as the link function.
	
	When exact balancing is feasible, the balanced weights by  entropy balancing   \citep{hainmueller2012entropy,zhao2016entropy}  have  the same expression as in  equation \eqref{MBweight}, though  parameter estimation is  different because it does not involve  regularization. In the bad-overlap or high-dimensional situation,  entropy balancing is infeasible. In contrast, 	there is no  normalization constraint $\sum_{i=1}^n T_i w_i=1$  in   the primal problem  \eqref{9}. Instead, we normalize the  weights   after solving  Problem \eqref{9}. This two-step strategy guarantees that Problem \eqref{9} is always feasible for any $\delta> 0$ under the mild condition that $W$ is positive-definite and $\phi_k(X_i)$ is bounded for any $i$ and $k$.

Moreover, by regularizing the dual parameter $\theta$ with the   $\ell_2$ norm,   Mahalanobis balancing stabilizes the estimated weights. The degree of regularization is determined by the tuning parameter $\delta$, which controls the level of residual multivariate imbalance after reweighting. In contrast to entropy balancing which  enforces  finite-sample univariate exact balance, Mahalanobis balancing maintains   finite-sample multivariate  approximate balance.

	When  the loss function  is  the quadratic function  by  \cite{zubizarreta2015stable}, the distance measure by \cite{chan2016globally}, or the Bregman distance by \cite{Josey2021}, the conjugate function can be calculated analogously but the expression can be  complicated. We refer the readers to the  monograph by \cite{Bertsekas2016} for an extensive discussion of the conjugate functions.  We prefer  the entropy loss function  $f(x)=x\log(x)$ in Mahalanobis balancing for its theoretical properties \citep{zhao2016entropy} and stable performance in our numerical experience.
	
	Intuitively, a larger value of $\delta$ results in more conservative $\theta$ value, which   leads to more stable     MB weights. Consequently,    treatment effect estimation may be more biased but exhibit less variability. Hence, selection of $\delta$ is a key to the bias and variance  trade-off for  treatment effect estimation. In Section \ref{sec3.2.2}, we make comparison of the proposed method to existing balancing methods from the perspective of treatment effect estimation. In Section \ref{sec3.3}, we study the asymptotic properties. In Section \ref{sec3.2}, we discuss  selection of  $\delta$ in details.
	
	\subsection{A comparison with existing balancing methods}\label{sec3.2.2}

	We compare  Mahalanobis balancing to  existing balancing methods from the outcome modelling perspective.  Write
	$${Y_i(1)} =   \mu_1(X_i) + \epsilon_{1i},$$
	where $\mu_1(X_i)=E\{Y_i(1)\mid X_i\}$. Assume that the $\epsilon_{1i}$ are mutually independent with mean zero and finite variance, and that they are independent of the covariates. Let $\hat{w}_1,\cdots, \hat{w}_n$ be the normalized weights estimated by an arbitrary weighting method, and  $\hat{\tau}_1=\sum_{i=1}^n T_i \hat{w}_i Y_i$ be the resulting weighted estimator of $\tau_1=E\{Y(1)\}$. It follows that \citep{wong2018kernel}
	$$
	\widehat{\tau}_1 -\tau_1 =  \underbrace{\sum_{i=1}^n \left(T_i\hat{w_i}- \frac{1}{n}\right)\mu_1(X_i)}_{A_1}    + \underbrace{\sum_{i=1}^n T_i\hat{w_i}\epsilon_{1i}}_{A_2}+ \underbrace{\frac{1}{n}\sum_{i=1}^{n}\mu_1(X_i)- E\{Y(1)\}}_{A_3}.
	$$
	
	A good weighting method should minimize or control the magnitudes of  the terms $A_1$ and $A_2$.  We first discuss the term $A_1$. If  the linear outcome model  $\mu_1(X) = \beta_1^{\top} \Phi(X)$ holds \citep[e.g.,][]{athey2016approximate}, then
	$A_1 = \beta_1^{\top}  (\sum_{i=1}^n T_i\hat{w_i}\Phi(X_i)- \bar{\Phi})$.
	The exact balancing  methods \citep[e.g.,][]{hainmueller2012entropy,imai2014covariate,chan2016globally,fan2022optimal} impose $\sum_{i=1}^n T_i \hat{w_i}\Phi(X_i) = \bar{\Phi}$, and thus  $A_1$ is exactly zero. In the bad-overlap or high-dimensional situation,  the term $A_1$ cannot be fully eliminated, and exact balancing is not applicable.   The covariate functional balancing method  \citep{wong2018kernel} directly minimizes $A_1^2$ with the basis functions $\Phi(\cdot)$  restricted to the RKHS.
	The MDABW method \citep{zubizarreta2015stable,wang2020minimal}  bounds the $k$th component of $|\sum_{i=1}^n T_i\hat{w_i}\Phi(X_i)- \bar{\Phi}|$ by a threshold parameter $\delta_k$, $k=1,\cdots, K$. Therefore, $|A_1|\leq |\beta_1^{\top}\delta|$, where $\delta=(\delta_1,\cdots, \delta_K)^\top$, suggesting that $A_1$ is  controlled by  the MDABW method. However, it is difficult to tighten the upper-bound by tuning $(\delta_1,\cdots, \delta_K)$ alone without modelling the outcomes.
	Note  that
	$|A_1|\leq\|\beta_1^\top {W}^{-1/2}\|_2  \sqrt{\text{GMIM}_1^{\hat{w}}}$, where the term $\text{GMIM}_1^{\hat{w}}$ does not involve the unknown regression parameter $\beta_1$.
	This inequality immediately suggests a simple guideline for the selection of the univariate threshold parameter $\delta$ in  Mahalanobis balancing. A good choice of  $\delta$ should make $\text{GMIM}_1^{\hat{w}}$ small enough so that the term $A_1$ is under control. We provide more details in Section \ref{sec3.2}.

	Moreover,  regardless of the linear outcome model assumption, the Cauchy-Schwarz inequality leads to that
	$|A_1|\leq
	\sqrt{\|\hat{w}\|_{2}^2-1/n}\sqrt{\sum_{i=1}^n \{\mu_1(X_i)\}^2}$
	and $|A_2|\leq \|\hat{w}\|_{2}\|\epsilon_1\|_{2}$, where  $\|\hat{w}\|_{2}=\sqrt{\sum_{i=1}^n T_i\hat{w}_i^2}$
	and  $\|\epsilon_1\|_{2}=\sqrt{\sum_{i=1}^n T_i\epsilon_{1i}^2}$. Therefore,  to minimize or control the magnitudes of  the terms $A_1$ and $A_2$, a good weighting method should control $\|\hat{w}\|_{2}$. Intuitively, it suggests that the weights should be  stable  and no extreme weights are allowed.  The covariate functional balancing method  \citep{wong2018kernel} uses $\|\hat{w}\|_{2}^2$ as a regularizer.
	The approximately residual balancing method \citep{athey2016approximate} minimizes a linear combination of  $\|\sum_{i=1}^n T_i\hat{w_i}\Phi(X_i)- \bar{\Phi}_0\|_{\infty}^2$ and $\|\hat{w}\|_{2}^2$, where $\bar{\Phi}_0$ is the sample average of $\Phi(X_i)$ in the control group.
	Mahalanobis balancing stabilizes the estimated weights via the  $\ell_2$ norm regularizer.  We  show that $\|\hat{w}\|_{2} = O_p(n^{-1/2})$ holds for  Mahalanobis balancing in the Supplementary Material.

	\subsection{Asymptotic properties}\label{sec3.3}
	The proposed Mahalanobis balancing method estimates the average treatment effect by
	$$\hat{\tau}^{MB}=\sum_{i=1}^{n} T_{i}\hat{w}_i^{MB} Y_{i}-\sum_{i=1}^{n}\left(1-T_{i}\right) \hat{w}_i^{MB} Y_{i},$$ where the $\hat{w}_i^{MB}$ are the estimated MB weights obtained by normalizing  the weights  $\hat{w}_i$ in equation \eqref{relation}.
	In this subsection, we prove that this MB-based ATE estimator is    doubly robust and semiparametrically efficient under mild regularity conditions. The proofs are given in the Supplementary Material.
	
	\begin{assumption} Assume the following conditions:
		
		(3.1). The optimization problem $\min _{{\theta} \in \Theta} E[ \tilde{f}^*(\theta^{\top} {W}^{1/2}(\Phi(X_i) - E\{\Phi(X)\}))]$ has  a unique global minimizer  for ${\theta}$, where ${\theta} \in \operatorname{int}(\Theta)$, and $\Theta$ is a compact parameter space for ${\theta}$.
		
		(3.2).  $\delta=o(n)$.
		
		(3.3).
		The conjugate function $\tilde{f}^{*}(\cdot)$ satisfies the property that if  $\tilde{f}^*(\theta^{\top} {W}^{1/2}(\Phi(X) - E\{\Phi(X)\}))=C*\tilde{f}^{*'}(\theta^{*\top} {W}^{1/2}(\Phi(X) - E\{\Phi(X)\}))$ for some  constant $C$ for all $X$, then  ${\theta}^{*}={\theta}$.
		
		(3.4). 	  $E\{\exp(a\theta^{\top}W^{1 / 2}[\Phi(X)-E\{\Phi(X) \mid T=j \} ])\} < \infty$ for all $\theta \in \Theta$, where $a = 1,2,3$, $j=0,1$. Moreover, $Var(\Phi(X)) < \infty$ and $Var(\Phi(X) \mid T)< \infty$.
		
		(3.5). The noises $\epsilon_{1i}={Y_i(1)}-E\{Y_i(1)\mid X_i\}$  are mutually independent  and sub-Gaussian. Similarly, the noises $\epsilon_{0i}={Y_i(0)}-E\{Y_i(0)\mid X_i\}$  are mutually independent  and sub-Gaussian. All the noises are independent of $X$.
		
		(3.6). $\eta \leq \pi(X) \leq 1-\eta$, where $0< \eta < 1/2$.
		
	\end{assumption}
	
	Assumption (3.1) is a standard requirement  for  consistency of the ATE estimator. Assumption (3.2) requires that the threshold parameter $\delta$ should not be  large. Assumption (3.3) is satisfied by  the common choices $f(x)=x\log(x)$ or $-\log(x)$.  Assumptions (3.4) and (3.5) are mild conditions to control the moments of the covariates and the noises. Assumption (3.6) is the strict overlap assumption.
	
	\begin{theorem}\label{thm2}
		
		Suppose that  Assumptions 1-3 hold, then the MB-based ATE estimator $\hat{\tau}^{MB}$ is doubly robust in the sense that:
		
		(i). If the propensity score satisfies $1/\pi(X)= C*\tilde{f}^{*'}(\theta^{*\top} {W}^{1/2}(\Phi(X)-\bar{\Phi}))$ for some $\theta^{*} \in \operatorname{int}(\Theta)$ and some constant $C$ for all $X$, then  $\widehat{\tau}^{MB} - \tau = O_p(n^{-1/2})$;
		
		(ii). If the conditional potential outcomes $E\{Y(1)|X\}$  and $E\{Y(0)|X\}$ are linear combinations of the basis functions $\Phi(X)=(\phi_1(X),\cdots, \phi_K(X))^\top$,    then  $\widehat{\tau}^{MB} - \tau = O_p(n^{-1/2})$ when at least one of the following two statements  is true: (a).  The loss function is the entropy function; (b).   The following conditions hold: \\$\sum_{i=1}^n I(T_i=j) (\hat{w_i}^{MB})^{3} / ({\sum_{i=1}^n I(T_i=j)(\hat{w}_i^{MB})^2})^{3/2} = o_p(1)$, and ${\sum_{i=1}^n I(T_i=j)(\hat{w}_i^{MB})^2} = O_p(n^{-1})$, $j=0,1$.
		
	\end{theorem}
	Note that the condition $1/\pi(X)= C*\tilde{f}^{*'}(\theta^{*\top} {W}^{1/2}(\Phi(X)-\bar{\Phi}))$ implies that the unnormalized MB weight in equation \eqref{relation} is proportional to the inverse weight $1/\{n\pi(X_i)\}$.  When we use entropy as the loss function, then this assumption says that
	$1/\pi(X)= C*\exp(\theta^{*\top} \Phi(X))$ holds for some   $\theta^{*}$  and constant $C$. The  two conditions in  statement (b)  require that  the weights should be smoothly distributed and no extremely large weights are allowed. That is, the  weights should be stable.   When the loss function is the entropy function, these two conditions    are automatically satisfied under   Assumption 3.

	\subsection{Tuning parameter selection}\label{sec3.2}
	
Recall that Problem \eqref{9} is always feasible for any $\delta> 0$ under  mild conditions. Therefore, we are allowed to freely try a set of grid points  and select the best $\delta$ that minimizes multivariate imbalance. In contrast to the MDABW method,  no resampling is required  in the selection procedure.
	
	Now, we present the  selection procedure  in details.
	First, set up a set of positive values $\mathcal{D}$ for $\delta$ selection. We use $\mathcal{D}=\{10^{-k}: k=0,1,\cdots, 6\}$ in the simulations and  application. In extensive numerical studies, we
find that searching $\delta$ in this range leads to satisfactory performance of the MB method. More grid points can  be added to $\mathcal{D}$ to potentially better control multivariate imbalance, though we do not find notable improvement in our simulations.

	The algorithm for selection of $\delta$ is as follows:
	
	\begin{algo} Selection of $\delta$:
		\label{algorithm 1}
		\begin{tabbing}
			\qquad \enspace For each $\delta\in \mathcal{D}$:\\
			\qquad \qquad Compute the dual parameter $\hat{\theta}$ by solving Problem \eqref{Dual};\\
			\qquad \qquad Compute the balanced  weights $\hat{w}_i$ using \eqref{relation};\\
			\qquad \qquad Obtain the Mahalanobis balancing weights by  normalization: $\hat{w}_i^{MB}={\hat{w}_i}/{\sum_{j:T_j=1}\hat{w}_j}$;\\
			\qquad \qquad Calculate  $\text{GMIM}_1^w$ in \eqref{GMD} using  the $\hat{w}_i^{MB}$;\\
			\qquad \enspace  Output $\delta^*$ that minimizes  $\text{GMIM}_1^w$.
		\end{tabbing}
	\end{algo}
	Similarly, we select the optimal threshold parameter for the control group by minimizing  $\text{GMIM}_0^w$.  Let $(\hat{w}_1^{MB},\cdots, \hat{w}_n^{MB})$ be the resulting MB weights corresponding to the  optimal threshold parameters in the two groups. The average treatment effect is then estimated by $\hat{\tau}^{MB}=\sum_{i=1}^{n} T_{i}\hat{w}_i^{MB} Y_{i}-\sum_{i=1}^{n}\left(1-T_{i}\right) \hat{w}_i^{MB} Y_{i}$.

	Interestingly, our empirical experience is that the ATE estimation is not quite sensitive to the choice of  $\delta$ if it is small enough. In the Supplementary Material,  the simulations reveal that the outputs  of ATE estimation and imbalance measure are similar when $\delta\leq 10^{-2}$. Therefore, if computation cost is a concern, we suggest  using  a fixed small value for $\delta$ in practise, say $\delta=10^{-4}$.

We offer some empirical guideline for evaluating whether the  $\text{GMIM}$ value is ``small enough'' for the MB method and other weighting methods.   The concept that a threshold for $\text{GMIM}$  is ``small enough'' is empirical,  as analogous to the  thresholds 0.1, 0.2 and  0.25 for ASMD \citep[e.g.][]{Stuart2013,imbens2015causal,Xie2019}. Our guideline is that if $\text{GMIM}\leq 1$ and $\text{GMIM}/K \leq 0.01$, then the covariates are approximately balanced in the multivariate sense.  The latter condition implies that the TASMD values are small. If either condition fails, then it implies that the  weighting method cannot effectively balance covariates, and the corresponding treatment effect estimation may exhibit significant bias.
	
	\subsection{High-dimensional Mahalanobis balancing}\label{sec3.4}
	In the high-dimensional setting where  $K$ is large compared to the sample size $n$, it becomes very difficult to simultaneously  balance all basis functions  even when each of the  basis functions is approximately balanced in the univariate sense. The difficulty is easily seen from  the multivariate  perspective: when the basis functions are mutually independent, the squared Mahalanobis distance increases linearly with $K$, and hence  multivariate  imbalance can be  hardly controlled for large $K$.   Exact balancing  is infeasible, and   it becomes  more difficult for univariate approximate balancing  to obtain optimal   threshold parameters.
	
	Mahalanobis balancing is still feasible in the high-dimensional situation under the mild conditions that the weight matrix is positive-definite and the basis functions are bounded. Moreover, it still only needs to tune one threshold parameter.
	By minimizing  multivariate imbalance to some extent, Mahalanobis balancing  generally produces more balanced weights and less biased treatment effect estimation compared to univariate approximate balancing.
	
	We provide asymptotic property for   the MB-based ATE estimator in the high-dimensional situation. We choose the entropy function as the loss function.
	
	\begin{assumption}  Assume the following conditions:
		
		(4.1). The number of constraints $K/n \rightarrow \tau $, where $0 < \tau < 1$.
		
		(4.2). $|\theta^{\top}W^{1/2}(E\{\Phi(X)\} - E\{\Phi(X)\mid T\})| \leq \gamma\log(n)$ for all $\theta \in \Theta$, where $\gamma>0$.
		
		(4.3). $\mu_j(X) = \beta_j^{\top}\Phi(X)$ with $\|\beta_j\|_2 = O(n^{\alpha})$, $j=0,1$, where $\alpha$ is a real number.
		
		(4.4) $\delta=O(n^{2s})$, where $s$ is a real number.
		
	\end{assumption}
	
	Assumption (4.1) restricts the  number of constraints $K$. It is allowed to increase with the sample size $n$, but it should not exceed $n$. Assumption (4.2) adds  additional assumption on $W^{1/2}(E\{\Phi(X)\} - E\{\Phi(X)\mid T\})$. It is satisfied, for instance, when  $\|\theta\|_2 \leq \gamma\log(n)$ and $\|E\{\Phi(X)\mid T\} - E\{\Phi(X)\}\|_2 \leq 1$, or when the $\theta$ is sparse.   Assumption (4.3) imposes linear outcome models where the magnitude of the coefficients is not too large.  Assumption (4.4) requires that $\delta$ should not be large.
	
	\begin{theorem}	\label{thm4}
		Suppose that  Assumptions 1-4 hold, then $\widehat{\tau}^{MB} - \tau = O_p(n^{-1/2}) + O_p(n^{s+ \gamma + \alpha})$.
	\end{theorem}
	
	It suggests that $\delta$ should be small, so one may set $s=0$ in Assumption (4.4). According to  Theorem \ref{thm4}, given that $s+ \gamma + \alpha\leq-1/2$, the MB-based ATE estimator $\widehat{\tau}^{MB}$ is $\sqrt{n}$-consistent for $\tau$. Moreover, this theorem  illustrates the difficulty of ATE estimation in the high-dimensional setting. For example, when $K = n$ and $\gamma > 0$, then Theorem \ref{thm4} no longer guarantees the consistency of the MB-based ATE estimator. In the Supplementary Material, we propose a modified version of Mahalanobis balancing,  and make extensive comparison of Mahalanobis balancing with existing regularized balancing methods in both sparse and non-sparse high-dimensional simulation settings.

\subsection{Semiparametric efficiency}
    Next, we prove the semiparametric efficiency property for the MB-based ATE estimator.

	\begin{assumption} Assume the following conditions:
		
		(5.1). Suppose that $1/\pi(X)= \tilde{f}^{*'}(m^{*}(X))$ for all $X$, where $m^{*}(\cdot) \in \mathcal{M}$ and  $\mathcal{M}$  is a set of smooth functions satisfying $\log n_{[]}\{\varepsilon, \mathcal{M}, L_{2}(P)\} \leq$ $C_1(1 / \varepsilon)^{1 / k_{1}}$. Here, $C_1$ is a positive constant, $k_{1} > 1/2$, and $n_{[]}\{\varepsilon, \mathcal{M}, L_{2}(P)\}$ denotes the covering number of $\mathcal{M}$ by $\varepsilon$-brackets.
		
		(5.2). $\tilde{f}^{*'}(m^{*}(X))$ is Lipschitz in $\mathcal{X}$, where $\mathcal{X}$ is the domain of covariate $X$.
		
		(5.3). There exists a value $\theta^{*}$ such that  $\sup_{x \in \mathcal{X}} \|\theta^{*\top} \left(\Phi\left(x\right)-\bar{\Phi}\right) -  m^{*}(x)\|_2 = o_p(1)$.
		
		(5.4).  The conditional potential outcomes $E\{Y(1)|X\}$  and $E\{Y(0)|X\}$ are linear combinations of the basis functions $\Phi(X)=(\phi_1(X),\cdots, \phi_K(X))^\top$.

        (5.5). There exists a constant $C_2$ such that $$\sup _{x \in \mathcal{X}}\|\left(\Phi\left(x\right)-\bar{\Phi}\right)\|_2 \leq C_2 n^{1 / 2} \text { and } E\left\{\|\left(\Phi\left(x\right)-\bar{\Phi}\right)\|_2^2\right\} \leq C_2$$

        (5.6). $s+ \gamma + \alpha < -\frac{1}{2}$.

        (5.7). Let $g(X_i) = \frac{T_i}{\pi(X_i)}\left(\Phi\left(X_{i}\right)-\bar{\Phi}\right)$, then $g(X)$ has independent $\sigma^2$-sub-gaussian entries $g_{j}$ with $\operatorname{Var}\left(g_{j}\right)>\tau^2$, where $\sigma=O(\operatorname{polyLog}(\mathrm{n}))$ and $\tau=O\left(\frac{1}{\operatorname{polyLog}(\mathrm{n})}\right)$.
	\end{assumption}
	
	Our setting is different from \cite{wang2020minimal}, where they require the dimension $K$ should be shrunk faster than the sample size $n$. Assumption (5.1) requires that  the complexity of the function class is sufficiently smooth. \cite{wang2020minimal} noted that the H\"{o}lder class with smoothness parameter $v$   with $v/K > 1/2$  satisfies this condition \citep[see also][]{van1996weak, fan2022optimal}.   Assumption (5.2) bounds the second derivative of the function $\tilde{f}^{*}(\cdot)$. Assumption (5.3) requires that the $m^*(\cdot)$ can be approximated by the linear span of the basis functions. It is similar to Assumption 1.6 by \cite{wang2020minimal}.  Assumption (5.4) is a standard condition  for  semiparametric efficiency. This condition is different from \cite{wang2020minimal}, since we discuss the weaker assumption that $K/n \rightarrow \tau>0$. Assumption (5.5) is a standard technical assumption that restricts the magnitude of the basis functions. Assumption (5.6) corresponds to the level of sparsity. Assumption (5.7) provides a lower bound for the smallest eigenvalue for the empirical covariance matrix of $\Phi(X)$.

	\begin{theorem}
		Suppose that Assumptions 1-5 hold,  then the  MB-based ATE estimator reaches the semiparametric efficiency bound.
	\end{theorem}

	\section{Numerical Studies}\label{sec4}
	In this section, we compare  Mahalanobis balancing to two classes of existing balancing methods for  ATE estimation in numerical studies. The first class consists of exact balancing methods, including entropy balancing (EB) \citep{hainmueller2012entropy}    implemented by  the R package \texttt{WeightIt}, covariate balancing propensity score (CBPS) \citep{imai2014covariate}  implemented by   the R package \texttt{CBPS}, and calibration weighting (CAL) \citep{chan2016globally}  implemented by the R package \texttt{ATE}. The second class consists of univariate approximate balancing methods, among which  the MDABW method \citep{wang2020minimal} is implemented by the R package \texttt{sbw}. We also report  the unadjusted ATE estimator using the simple difference in the outcomes (Unad) and the propensity score weighting estimator (PS) using logistic regression  implemented by the R package \texttt{WeightIt}. In the end of this section, we discuss the kernel-based covariate balancing method  \citep{wong2018kernel}.
	
	For each method, we report  the bias (Bias) of  ATE estimation, Monte Carlo standard deviation (SD), root mean squared error (RMSE), Monte Carlo average of the  mean of the $\text{TASMD}_{k}$, $k=1,\cdots, K$  (meanTASMD) as a univariate post-reweighting imbalance measure, and  the generalized  Mahalanobis imbalance measure (GMIM) $\text{GMIM}=\text{GMIM}_1+\text{GMIM}_0$   with  ${W}_1$  as the weight matrix  to quantify residual multivariate imbalance.
	
    We assess the performance of  Mahalanobis balancing  using  ${W}_1$  (denoted by MB) or ${W}_2$ (denoted by MB2) as the weight matrix, respectively, and make comparison with  exact balancing  methods (EB, CBPS, CAL) and the MDABW method. We consider  five scenarios.

	For each scenario, we conduct 1000 Monte Carlo simulations. The sample size is set to be $n=200$ or $n=1000$. Table \ref{Table 1} summarizes the outputs. The Supplementary Material provides additional information.

	\begin{table}[b]
		\caption{ Simulation outputs in the low-dimensional settings.}
		\label{Table 1}
		\resizebox{\textwidth}{!}{
			\begin{tabular}{ccccccccccc}
				\toprule
				& \multicolumn{1}{c}{{\textbf{Unad}}} & \multicolumn{1}{c}{\textbf{PS}} & \multicolumn{1}{c}{\textbf{EB}} &  \multicolumn{1}{c}{\textbf{CBPS}} & \multicolumn{1}{c}{\textbf{CAL}} & \multicolumn{1}{c}{\textbf{MB}} & \multicolumn{1}{c}{\textbf{MB2}} & \multicolumn{1}{c}{\textbf{MDABW}} & \multicolumn{1}{c}{\textbf{KERNEL}} \\
				\midrule
				\multicolumn{10}{l}{ Scenario A  with $(n,K)=(200,10)$: both propensity score and  outcome models are misspecified.}\\
				\multicolumn{1}{c}{{\textbf{Bias}}}
				&-1.79   &0.57   &0.31   &0.28   &0.31   &0.29   &0.29   &0.51  &-0.03\\
				\multicolumn{1}{c}{\textbf{SD}}
				&2.94 &3.33 &3.06 &3.12 &3.06 &3.06 &3.06 &3.08 &3.09\\
				\multicolumn{1}{c}{\textbf{RMSE}}
				&3.44 &3.38 &3.08 &3.13 &3.08 &3.07 &3.07 &3.12 &3.09\\
				\multicolumn{1}{c}{{\textbf{meanTASMD}}}
				&0.15 &0.06 &0.00 &0.05 &0.00 &0.00 &0.00 &0.01 &0.02\\
				\multicolumn{1}{c}{{\textbf{GMIM}}}
				&0.20 &0.06 &0.00 &0.02 &0.00 &0.00 &0.00 &0.00 &0.01\\
				\\
				\multicolumn{10}{l}{  Scenario B   with $(n,K)=(200,65)$:   the confounders include covariate interactions.}\\
				\multicolumn{1}{c}{{\textbf{Bias}}}
				&6.93 &4.14 &- &3.72 &- &-0.53 &3.13 &1.44 &2.95\\
				\multicolumn{1}{c}{\textbf{SD}}
				&4.07 &7.49 &- &5.17 &- &1.00 &2.58 &1.38 &2.11\\
				\multicolumn{1}{c}{\textbf{RMSE}}
				&8.04 &8.56 &- &6.37 &- &1.13 &4.05 &1.99 &3.62\\
				\multicolumn{1}{c}{{\textbf{meanTASMD}}}
				&0.21 &0.33 &- &0.23 &- &0.04 &0.13 &0.12 &0.12\\
				\multicolumn{1}{c}{{\textbf{GMIM}}}
				&2.09 &7.36 &- &2.94 &- &0.12 &0.81 &0.55 &225.50\\
				\\
				\multicolumn{10}{l}{   Scenario C  with $(n,K)=(200,10)$:    covariate means are very different.}\\
				\multicolumn{1}{c}{{\textbf{Bias}}}
				&15.00 &8.71 &- &7.05 &- &0.69 &0.92 &3.31 &8.94 \\
				\multicolumn{1}{c}{\textbf{SD}}
				&0.76 &4.24 &- &4.52 &- &0.84 &0.92 &1.45 &2.12\\
				\multicolumn{1}{c}{\textbf{RMSE}}
				&15.01 &9.69 &- &8.37 &-  &1.09 &1.30 &3.61 &9.19\\
				\multicolumn{1}{c}{{\textbf{meanTASMD}}}
				&0.89 &0.67 &- &0.51 &- &0.04 &0.05 &0.20 &0.53\\
				\multicolumn{1}{c}{{\textbf{GMIM}}}
				&4.07 &2.80 &- &1.74 &- &0.03 &0.04 &0.23 &1.63\\
				\\
				\multicolumn{10}{l}{  Scenario D  with $(n,K)=(1000,10)$:    sample sizes are very imbalanced across treatment groups.}\\
				\multicolumn{1}{c}{{\textbf{Bias}}}
				&-8.60 &-1.60 &- &-0.74 &- &-0.25 &-0.25 &-1.37 &-6.06 \\
				\multicolumn{1}{c}{\textbf{SD}}
				&0.31 &2.11 &- &0.41 &- &0.44 &0.44 &0.77 &2.32\\
				\multicolumn{1}{c}{\textbf{RMSE}}
				&8.60 &2.65 &- &0.85 &- &0.50 &0.50 &1.57 &6.49\\
				\multicolumn{1}{c}{{\textbf{meanTASMD}}}
				&0.48 &0.25 &- &0.19 &- &0.01 &0.01 &0.09 &0.31\\
				\multicolumn{1}{c}{{\textbf{GMIM}}}
				&1.58 &0.86 &- &0.28 &- &0.00 &0.00 &0.05 &1.03\\
                \\
                \multicolumn{10}{l}{  Scenario E with $(n,K)=(200,5)$:   heavy-tailed covariates.}\\
				\multicolumn{1}{c}{{\textbf{Bias}}}
				&-8.91 &-3.40 &- &-1.33 &- &-0.66 &-0.66 &-2.55 &-4.41\\
				\multicolumn{1}{c}{\textbf{SD}}
				&0.61 &2.10 &- &0.85 &- &0.68 &0.68 &1.56 &1.12\\
				\multicolumn{1}{c}{\textbf{RMSE}}
				&8.93 &3.99 &- &1.58 &- &0.94 &0.95 &2.99 &4.55\\
				\multicolumn{1}{c}{{\textbf{meanTASMD}}}
				&0.35 &0.25 &- &0.23 &- &0.02 &0.02 &0.16 &0.14\\
				\multicolumn{1}{c}{{\textbf{GMIM}}}
				&0.91 &0.70 &- &0.42 &- &0.01 &0.01 &0.21 &0.21\\
				\bottomrule
			\end{tabular}
		}
	\end{table}

	In Scenario A, we consider the situation where  both propensity score and outcome models are misspecified. This setting similar to that of  \cite{kang2007demystifying}.  For each simulation, we first generate a standard normal random vector $Z=\left(Z_{1}, \ldots, Z_{p}\right)^{\top}$ with $p=10$ for each observation, and then generate the  covariates ${X}$ from $X_{1}=\exp(Z_{1} / 2), X_{2}=Z_{2}/(1+\exp (Z_{1})), X_{3}=(Z_{1} Z_{3} +0.6)^{3}, X_{4}=(Z_{2}+Z_{4}+20)^{2}, X_{j}=Z_{j}, j=5, \cdots, 10$. The potential outcomes are derived from  linear regression  utilizing the variables $Z$ as follows: $(Y(1), Y(0))=(210+ 13.7\sum_{i=1}^{4}Z_{i}+\epsilon_1, 210-6.85\sum_{i=1}^{4}Z_{i}+\epsilon_0)$, where $\epsilon_1$ and $\epsilon_0$ are independent standard normal variables which are independent of  $(T,Z)$. The true ATE is zero. The treatment indicator is generated from  $T\sim \text{Bernoulli}(\pi(Z))$, where $\pi(Z)= 1/(1+\exp(0.5Z_{1}+0.1Z_{4}))$.  The $Z$ are latent variables and the $X$ are observed. Hence,  when one specifies a linear outcome model and a logistic propensity score model using the $X$ as the covariates, both models are misspecified. In this scenario, $\Phi(X)=X$, and  $K=p=10$. Scenario A is a  good-overlap setting, since the distributions of $X$ are not quite different  between the two treatment groups, and   the exact balancing methods are applicable.
	
	The MB and MB2 methods and  other balancing methods show similar outputs.  They all successfully remove covariate imbalance and obtain almost identical ATE estimates. The RMSEs of MB, MB2, CAL, and EB are somewhat smaller than those of MDABW and CBPS. We conclude that Mahalanobis balancing  maintains the advantages of the exact balancing methods in Scenario A. Moreover, the computation cost of Mahalanobis balancing is less than $0.5\%$ than that of MDABW. This is because  MDABW   needs 1000 bootstraps for each possible  tuning parameter.

	In Scenario B,  we first generate $T\sim \text{Bernoulli}(0.5)$ for each observation. Set $p=10$. If $T=1$,  the ten-dimensional covariates are generated from ${X}\sim N({1}, {\Sigma_1})$, where   $Cov(X_j,X_k)=2^{-I(j \neq k)}$; if $T=0$, then the covariates are   generated from ${X}\sim N({1}, {\Sigma_0})$, where  $Cov(X_j,X_k)=I(j = k)$.  This data generation procedure allows us to delineate the discrepancy of covariate distributions between the treated and control groups: the  mean and variance  are the same,  but the interaction terms are  different. The potential outcomes are generated from  $(Y(1), Y(0))=(2\sum_{i=1}^{10}X_{i} + 2(X_{1}X_{2}+X_{2}X_{3}+\cdots+X_{9}X_{10}+X_{10}X_{1})+\epsilon_1, \sum_{i=1}^{10}X_{i} + (X_{1}X_{2}+X_{2}X_{3}+\cdots+X_{9}X_{10}+X_{10}X_{1})+\epsilon_0)$,
	where $\epsilon_j\sim N(0,1)$, $j=0,1$. The true value of ATE is calculated via simulation.  The first two  moments of $X$ are well balanced. In contrast, the interaction terms are strong confounders. We set $\Phi(X)=\{X_i, X_jX_k: 1 \leq i \leq 10, 1\leq j \leq k\leq 10\}$ and thus $K=65$.
	
	Because the covariance structures are  quite dissimilar in the treated and control groups,  both EB and CAL do not admit solutions, and thus  Scenario B is a bad-overlap setting. In comparison, MB is able to approximately balance the interaction terms.  The two imbalance measures of  MB  are substantially lower than those of MDABW and CBPS. Furthermore, MB  significantly improves the performance of  exact balancing and univariate approximate balancing for ATE estimation. It exhibits much smaller bias and RMSE than those of MDABW and CBPS. Nevertheless, MB2 performs unsatisfactorily because  the weight matrix $W_2$ is unstable due to considerable correlation among the interaction terms, and its  GMIM value is much larger than that of MB. We do not recommend MB2 in this situation.
	
	In Scenario C, the covariate means   are  different  in the two groups.  Set $p=10$. We first generate $T\sim \text{Bernoulli}(0.5)$ for each observation. If $T=1$,  the covariates are    generated from ${X}\sim N({1}, {\Sigma_1})$; otherwise,  ${X}\sim N({0}, {\Sigma_0})$. The covariance matrices $\Sigma_1$ and $\Sigma_0$ are the same as those in Scenario B. The potential outcomes are generated from  $(Y(1), Y(0))=(2\sum_{i=1}^{10}X_{i} +\epsilon_1, \sum_{i=1}^{10}X_{i} +\epsilon_0)$,
	where $\epsilon_j\sim N(0,1)$, $j=0,1$. The true ATE is 5.  In this scenario, we set $\Phi(X)=X$, and thus $K=10$. Scenario C is  a bad-overlap setting, because  the covariate distributions are quite different such that  EB and CAL are infeasible.   MDABW exhibits much less bias compared to CBPS, but it has larger standard deviation. The proposed MB and MB2 methods have  least biased ATE estimation,  smallest  RMSEs, and  smallest  GMIM values.
	
	In Scenario D,  the sample sizes are highly imbalanced. This situation is commonly seen in  cohort studies where the exposure is rare and the size of the control group  is very large.  We fix the expected  sample size of the treated group to be 50, and  the control group has a  sample size of roughly 950.  Set   $\Phi(X)=X$ and  $K=p=10$. The covariates are simulated by $X\sim N({1}, {I}_{10\times 10})$. The treatment indicator is simulated from   $T\sim Bernoulli(\pi(X))$ with $\pi(X)=1/\{1+19  \exp(\sum_{i=1}^{10}X_{i}-10)\}$. The potential outcomes are simulated from  $(Y(1), Y(0))=(2\sum_{i=1}^{10}X_{i} +\epsilon_1, \sum_{i=1}^{10}X_{i} +\epsilon_0)$,
	where $\epsilon_j\sim N(0,1)$, $j=0,1$.  The true ATE is 10. Again, this is a bad-overlap setting, where EB and CAL are infeasible. In contrast to Scenario C,  the ATE estimate of CBPS is less biased than that of MDABW, and the RMSE  is  smaller. MB and MB2 greatly outperform MDABW and CBPS in terms of bias, RMSE, and GMIM.

    In Scenario E, the distribution of the covariates is heavy-tailed. In specific, the covariate $X$ is generated by  $\log(X)\sim N({0}, {I}_{5\times 5})$. The treatment indicator is generated from $T\sim Bernoulli(\pi(X))$, where $\pi(X)=1/\{1+0.1 \exp(\sum_{i=1}^{5}X_{i}-5)\}$. The potential outcomes are simulated from  $(Y(1), Y(0))=(2\sum_{i=1}^{5}X_{i} +\epsilon_1, \sum_{i=1}^{5}X_{i} +\epsilon_0)$,
	where $\epsilon_j\sim N(0,1)$, $j=0,1$.  The true ATE is 8.24. Once again, this is a bad-overlap scenario, where the EB and CAL methods are not feasible. The CBPS estimate  is less biased than that of  MDABW, and the RMSE is smaller. Moreover, the MB and MB2 significantly outperform other methods.

	We remark that although CBPS is an exact balancing method, it is still applicable to Scenarios B, C,  D, and E, all of which are  bad-overlap settings. This is because it employs the generalized method of moments for parameter estimation, which still works even when the  moment constraints are not met exactly. In this sense, CBPS can be regarded as an approximate balancing method, though it does not control univariate dispersion of each covariate. The CBPS method has larger meanTASMD and GMIM values than the MDABW method in all scenarios.  Scenarios B, C,  D, and E suggest that the covariate distributions in each  group (that is, the conditional distributions of $X|T=1$ and $X|T=0$)  are still highly dissimilar to the distribution of $X$ in the population  after CBPS reweighting, leading to large estimation bias and RMSE values.

	The GMIM recognizes post-weighting multivariate imbalance  and is  predictive of the performance of balancing methods for treatment effect estimation. Therefore, we recommend to use GMIM instead of meanTASMD to assess residual covariate imbalance in the bad-overlap settings.

    We also implement the kernel-based covariate balancing method by \cite{wong2018kernel}, and the results are in column ``KERNEL'' in Table \ref{Table 1}. In Scenario A where the   model  setup is  similar to that in \cite{wong2018kernel}, the kernel-based covariate balancing method outperforms other methods. This result is not surprising because the nonparametric kernel method tends to be more robust to model misspecification. However, in all the bad-overlap scenarios (Scenarios B, C,  D, and E), the performance of the kernel-based method is not satisfactory. The slow convergence rate of the nonparametric method may explain the inferior performance of the  kernel-based method in  finite samples. Hence, while the kernel-based covariate balancing method may be preferred in some cases where model misspecification is a concern, it may not always be the optimal choice for achieving balance in observational studies especially when the overlap between two groups is poor, the sample sizes are imbalanced, or  the covariates are heavy-tailed.

	\subsection{Application}\label{sec5}

	\begin{table}[b]\center
		\caption{Data analysis of National Supported Work program.}
		\label{Table 3}
        \resizebox{\textwidth}{!}{
		\begin{tabular}{cccccccccc}
			\toprule
			& \multicolumn{1}{c}{{\textbf{Unad}}} & \multicolumn{1}{c}{\textbf{PS}} & \multicolumn{1}{c}{\textbf{EB}} &  \multicolumn{1}{c}{\textbf{CBPS}} & \multicolumn{1}{c}{\textbf{CAL}} & \multicolumn{1}{c}{\textbf{MB}} & \multicolumn{1}{c}{\textbf{MB2}} & \multicolumn{1}{c}{\textbf{MDABW}} & \multicolumn{1}{c}{\textbf{KERNEL}} \\
			\midrule
			\textbf{ATE} &-606.09 &635.29 &- &1329.66 &- &1529.37 &1516.03 &960.71 &141.87 \\
			\textbf{SE}
			&690.03 &1368.81 &- &1525.42 &- &2044.52 &2082.54 &1231.93 &1018.23\\
			\textbf{maxTASMD} &1.31 &0.35 &- &0.38 &- &0.01 &0.02 &0.12 &0.26\\
			\textbf{medTASMD} &0.34 &0.21 &- &0.19 &- &0.00 &0.02 &0.08 &0.12\\
			\textbf{meanTASMD} & 0.56 &0.21 &- &0.20 &- &0.00 &0.01 &0.08 &0.11\\
			{\textbf{GMIM}}
			&3.38 &0.49 &- &0.30 &- &0.00 &0.01 &0.05 &0.21\\
			\bottomrule
		\end{tabular}
        }
	\end{table}
	
	We revisit a dataset from the National Supported Work program \citep{dehejia1999causal}. The National Supported Work program is a labor training program implemented in the 1970s by providing work experience to  selected individuals.   The data consist of a National Supported Work experimental group with sample size $185$ and a nonexperimental comparison group from the Panel Study of Income Dynamics with sample size $429$. We regard them as the treated and control groups, respectively. The outcome variable $Y$ is the post-intervention earning measured in Year 1978.
	
	The covariates include four numeric covariates $age$, $education$, $earn1974$, $earn1975$, four binary variables  $married$, $black$, $nodegree$ and $hispanic$, and  sixteen interaction terms between the numeric covariates and the binary variables. In addition, we include the ratio $education/age$, which represents  possible nonlinear effect of $age$ at each level of $education$. Table  \ref{Table 3} shows the results. Here, maxTASMD and medTASMD are the Monte Carlo averages of the maximum and medium of $\text{TASMD}_{k}$, $k=1,\cdots, K$, respectively. 500 bootstrap resamples are used to calculate the standard error for all methods.
	
	This is a bad-overlap setting, since the EB and  CAL methods  fail to output ATE estimation. This can be  explained by  that the distribution of the covariate $education/age$ is highly disimilar between the two groups. The CBPS method substantially reduces univariate and multivariate imbalance compared to the unadjusted method. The PS method has similar performance. The MDABW and KERNEL methods perform better than CBPS in the sense of smaller residual univariate and multivariate imbalance after reweighting. In comparison, the proposed Mahalanobis balancing methods   produce most balanced weights as they have smallest TASMD and GMIM values, although they have  larger standard error. The Mahalanobis balancing methods are recommended when  balanced weights are in demand in  data analysis.

	\section{Concluding Remarks}
	In this paper, we proposed Mahalanobis balancing. It produces balanced weights by solving a convex optimization problem with a  quadratic constraint that represents  multivariate imbalance. We discussed the high-dimensional setting. Different choices of the basis functions were examined in the simulations. We compared  Mahalanobis balancing with  the exact balancing,  univariate approximate balancing, and  high-dimensional regularized balancing methods in extensive numerical studies, and found that Mahalanobis balancing generally led to more balanced weights and less biased ATE estimation.
	
	The proposed Mahalanobis balancing methods can be easily extended to the situation with multiple treatment arms. Other causal estimands, such as  average treatment effect on the treated and  average treatment effect on the control, can be also estimated using  Mahalanobis balancing, though some  modification is needed. For example, when the average treatment effect on the control is of interest, one needs to replace $\bar{\Phi}$ with the sample average of $\Phi(X)$ in the control group in Problem \eqref{9} to obtain Mahalanobis balancing  weights.
	
	Our method can be used to address the transportability and generalizability issues in data integration and data fusion.  Recently,  exact balancing methods were applied to these interesting problems \citep[e.g.][]{lee2023improving,Josey2021}. Nevertheless, the trial population and the target population usually exhibit high degree of discrepancy, and there can be  a large number of confounders. In these  situations,   exact balancing  is infeasible and Mahalanobis balancing  is recommended.
	
In the numerical studies, we  used entropy as the loss function. One may further investigate the performance of Mahalanobis balancing using other loss functions \citep[e.g.][]{chan2016globally,Josey2021}.  Moreover, one may consider other multivariate imbalance measures  instead of the Mahalanobis metric to improve robustness, e.g., energy distance \citep{HulinMak2020} and  kernel distance \citep{zhu2018kernel}. However, it is beyond the scope of this paper and we do not further pursue it here.

	\section*{Supplementary Material}
	The online Supplementary Material includes discussion of the high-dimensional setting,  proofs and additional simulation studies.

\section*{Acknowledgements}
Ying Yan's research is  supported by the National Natural Science Foundation of China (Grant No. 11901599).

	\bibliographystyle{biom.bst}
	\begin{spacing}{1.2}
		\bibliography{paper-ref.bib}
	\end{spacing}
 
\vspace{1cm}

\noindent \author{Ying Yan \\
\thanks{Corresponding Author. \href{mailto: yanying7@mail.sysu.edu.cn}{yanying7@mail.sysu.edu.cn}} \\
School of Mathematics, Sun Yat-sen University, Guangzhou, China, 510275\\
}

\end{document}


\def\spacingset#1{\renewcommand{\baselinestretch}%
		{#1}\small\normalsize} \spacingset{1}

	
	\if1\blind
	{
		\title{\bf Supplementary Material for ``Mahalanobis balancing: a multivariate perspective on approximate covariate balancing''}
		\author{Yimin Dai \hspace{.2cm}\\
			School of Mathematics, Sun Yat-sen University, Guangzhou, China, 510275\\
			and \\
			Ying Yan\thanks{Corresponding Author. \href{mailto: yanying7@mail.sysu.edu.cn}{yanying7@mail.sysu.edu.cn}} \\
			School of Mathematics, Sun Yat-sen University, Guangzhou, China, 510275 \\
		}
\date{}
		\maketitle
	} \fi
	
	\if0\blind
	{
		\bigskip
		\bigskip
		\bigskip
		\begin{center}
			{\LARGE\bf Supplementary Material for ``Mahalanobis balancing: a multivariate perspective on approximate covariate balancing''}
		\end{center}
		\medskip
	} \fi
	
	\bigskip
	
	\appendix

	\spacingset{1.5} 

	\section{High-dimensional Mahalanobis balancing}\label{sec3.4}
	
	Asymptotic theory of Mahalanobis balancing becomes much more difficult  in the ultra high-dimensional setting with $K\gg n$. For example, the general high-dimensional regularized M-estimation theory  \citep{Wainwright2019}  does not directly apply  because the   $\ell_2$ norm regularizer  in the dual problem (8) is not decomposable.  
	
	When the true propensity score model is not sparse,  Mahalanobis balancing outperforms  the   high-dimensional regularized balancing methods  \citep{athey2016approximate, ning2020robust,tan2020model,tan2020regularized} and  univariate approximate balancing   \citep{wang2020minimal}  in the simulation studies. Nevertheless, when the true propensity score model is sparse, the performance of Mahalanobis balancing is not entirely satisfactory compared to high-dimensional regularized balancing. This is not surprising because  high-dimensional regularized balancing   commonly selects a sparse subset of basis functions by $\ell_1$ or elastic net  regularization. In contrast, the   $\ell_2$ regularization in the dual problem (8) suggests that Mahalanobis balancing  does not automatically perform variable selection.
	
	In the sparse setting, we propose   the high-dimensional Mahalanobis balancing method, which is a modified version of Mahalanobis balancing. It differs from  Mahalanobis balancing by performing selection of basis functions before producing balanced weights.
	
	Recall that we write $X=(X_{1}, \cdots, X_{p})^\top$. In the following high-dimensional setting, we use the trivial feature mapping: $\Phi(X)=X$, and thus $K=p$. In high-dimensional Mahalanobis balancing, we select a subset  of the covariates $X$ with dimension $K_0$, where $K_0\leq p$. The following algorithm gives a principled way  for subset selection. Let $ASMD_j$ be the  absolute standardized mean difference for $X_j$, $j=1,\cdots, p$. Rank $X_{1}, \cdots, X_{p}$ as  $X_{(1)}, \cdots, X_{(p)}$ such that $X_{(1)}$ has the largest ASMD value, $X_{(2)}$ has the second largest ASMD value, and so forth.
	
	\begin{algo}  High-dimensional Mahalanobis balancing
		\label{algorithm 1}
		\begin{tabbing}
			\qquad  For each step $j\in\{1,\cdots,p\}$: \\
			\qquad \qquad  Apply Algorithm 1 to obtain  the MB weights $\hat{w}^{*MB}_i$ using   $\Phi(X) = (X_{(1)},\cdots, X_{(j)})^\top$; \\
			\qquad \qquad  Calculate  $\text{GMIM}_1^w$ in
			(6) using  the $\hat{w}_i^{MB}$, and define $\text{GMIM}_{1,j}^w=\text{GMIM}_1^w/j$; \\
			\qquad \qquad Add $(j,\text{GMIM}_{1,j}^w)$ to the scatter  plot;\\
			\qquad \qquad Observe whether there is a kink at  $(j,\text{GMIM}_{1,j}^w)$:\\
			\qquad \qquad \qquad  If no, let $j=j+1$;\\
			\qquad \qquad \qquad  If yes,  stop and output $K_0=j-1$.\\
			\qquad If no kink is observed, output $K_0=p$.
		\end{tabbing}
	\end{algo}
	We then obtain $\Phi(X)=(X_{(1)},\cdots, X_{(K_0)})^\top$ and   the MB  weights at Step $K_0$. Similarly, we obtain the  MB weights for the control group.  The average treatment effect is then estimated by the weighted difference of the outcomes.
	
	The rationale of Algorithm S1 hinges on the adjusted multivariate imbalance measure $\text{GMIM}_{1,j}^w$. It represents the average contribution of the $j$ most univariate imbalanced covariates  to the residual multivariate imbalance after MB weighting. If it remains stable as $j$ increases, then MB is  capable of controlling  multivariate imbalance. However, if there is a kink at  Step $j$, it implies that   adding the $j$th largest imbalanced covariate greatly increases the multivariate imbalance measure, implying that  MB starts to lose control of overall imbalance at Step $j$.  Therefore, we  stop and  choose the outputs at Step $j-1$. The kink  usually occurs when $K_0= O(\sqrt{p})$ in the numerical studies. If there is no   kink for all $j=1,\cdots, p$, it suggests that the performance of MB is acceptable even  if all $p$ covariates are included. One may choose $K_0=p$  in this situation, or $K_0=\sqrt{p}$ if a small subset of covariates is preferred. Our numerical experience  reveals that   high-dimensional Mahalanobis balancing  substantially improves  the performance of Mahalanobis balancing  in the sparse setting.
	
	Algorithm S1 tends to select covariates that exhibit large univariate imbalance.  Algorithm S1 may not work well when the unselected imbalanced covariates   are correlated with the outcomes.   To alleviate this issue, one may construct a bias-corrected version for treatment effect estimation by augmenting  the MB-based ATE estimator with the  outcome models.
	
	\subsection{Simulation for the high-dimensional settings}

	\begin{table}[t]
		\caption{Simulation outputs in the high-dimensional settings.}
		\label{Table 2}
		\resizebox{\textwidth}{!}{
			\begin{threeparttable}
				\begin{tabular}{ccccccccccc}
					\toprule
					& \multicolumn{1}{c}{{\textbf{Unad}}} & \multicolumn{1}{c}{\textbf{MB}} & \multicolumn{1}{c}{\textbf{kernelMB}}& \multicolumn{1}{c}{\textbf{hdMB}} & \multicolumn{1}{c}{\textbf{MDABW}} &
					\multicolumn{1}{c}{\textbf{RCAL1}} & \multicolumn{1}{c}{\textbf{RCAL2}}&
					\multicolumn{1}{c}{\textbf{ARB}} &
					\multicolumn{1}{c}{\textbf{hdCBPS}} \\
					\midrule
					\multicolumn{8}{l}{ Scenario  F with $(n,p)=(200,100)$:  sparse propensity score model.}\\
					\multicolumn{1}{c}{{\textbf{Bias}}}
					&-3.44&-0.85&-0.92&{-0.01}&-2.13&-0.48&-0.47&-0.22&-0.19
					\\
					\multicolumn{1}{c}{\textbf{SD}}
					&0.36&0.30&0.29&0.32&1.11&0.32&0.31&{0.27}&0.37
					\\
					\multicolumn{1}{c}{\textbf{RMSE}}
					&3.47&0.90 &0.96&{0.32}&2.39&0.58&0.57&0.35&0.38
					\\
					\multicolumn{1}{c}{\textbf{meanASMD}}
					&0.16&{0.07} &0.08&0.14&0.15&0.15&0.16&-&-
					\\
					\multicolumn{1}{c}{\textbf{GMIM}}
					&4.35&{0.53} &0.71&$4.19 (0.00)$\tnote{*}&2.12&2.17&3.43&-&-
					\\
					
					\\
					\multicolumn{8}{l}{ Scenario F with $(n,p)=(200,500)$:  sparse propensity score model.}\\
					\multicolumn{1}{c}{{\textbf{Bias}}}
					&-3.46&-2.46&-2.49&{0.00}&-3.43&-0.76&-0.76&-0.43&-0.38\\
					\multicolumn{1}{c}{\textbf{SD}}
					&0.37&0.33&0.33&0.35&0.37&0.29&0.29&{0.25}&0.40\\
					\multicolumn{1}{c}{\textbf{RMSE}}
					&3.48 &2.48&2.52&{0.35}&3.45&0.81&0.81&0.48&0.55
					\\
					\multicolumn{1}{c}{\textbf{meanASMD}}
					&0.12&{0.10} &0.10&0.22&0.12&0.12&0.13&-&-\\
					
					\multicolumn{1}{c}{\textbf{GMIM}}
					&7.32&{5.15} &5.30&$36.26(0.00)$&7.28&7.36&7.43&-&-
					\\
					\\
					\multicolumn{8}{l}{ Scenario G with $(n,p)=(200,100)$:  dense propensity score model.}\\
					\multicolumn{1}{c}{{\textbf{Bias}}}
					&-1.35&-0.35&{-0.31}&-0.55&-0.91&-0.76&-0.75&-0.54&-0.86\\
					\multicolumn{1}{c}{\textbf{SD}}
					&0.21&0.23&{0.22}&0.39&0.39&0.30&0.30&0.23&0.23\\
					\multicolumn{1}{c}{\textbf{RMSE}}
					&1.36&0.42 &{0.38}&0.68&0.99&0.82&0.81&0.59&0.89\\
					\multicolumn{1}{c}{\textbf{meanASMD}}
					&0.20&{0.07} &0.08&0.16&0.15&0.20&0.20&-&-\\
					\multicolumn{1}{c}{\textbf{GMIM}}
					&3.53&{0.64} &0.82&$4.87(0.00)$&1.86&3.03&3.20&-&-
					\\
					\\
					\multicolumn{8}{l}{ Scenario G  with $(n,p)=(200,500)$:  dense propensity score model.}\\
					\multicolumn{1}{c}{{\textbf{Bias}}}
					&-0.33&-0.23&{-0.22}&-0.27&-0.33&-0.32&-0.33&-0.27&-0.34\\
					\multicolumn{1}{c}{\textbf{SD}}
					&0.17&0.17&0.17&0.33&{0.17}&0.18&0.18&0.19&0.17\\
					\multicolumn{1}{c}{\textbf{RMSE}}
					&0.37&0.29 &{0.28}&0.43&0.37&0.37&0.37&0.33&0.39
					\\
					\multicolumn{1}{c}{\textbf{meanASMD}}
					&0.12&{0.10} &0.10&0.22&0.12&0.18&0.13&-&-\\
					\multicolumn{1}{c}{\textbf{GMIM}}
					&7.29&{5.13} &5.28&$34.55(0.00)$&7.26&7.38&7.52&-&-\\
					\bottomrule
				\end{tabular}
				\begin{tablenotes}
					\footnotesize
					\normalsize
					
					\item[*] If we use $\Phi(X)=X$ to compute the GMIM for the hdMB method, then $\text{GMIM}^w=4.19$; if we use $\Phi(X)=(X_{(1)},\cdots, X_{(K_0)})^\top$ with $K_0$ determined
					by Algorithm S1, then $\text{GMIM}^w=0.00$.	
				\end{tablenotes}
			\end{threeparttable}
		}
		
	\end{table}

	In this subsection, we assess the performance of three Mahalanobis balancing methods in the high-dimensional setting, and compare them with high-dimensional regularized balancing methods, including
	two versions of regularized calibrated estimation  (RCAL1 and RCAL2) \citep{tan2020model,tan2020regularized}  implemented with the R package \texttt{RCAL},  approximately  residual balancing (ARB) \citep{athey2016approximate} implemented with the R package \texttt{balanceHD}, and high-dimensional covariate balancing propensity score (hdCBPS) \citep{ning2020robust}  implemented with the R package \texttt{CBPS}. We also report  the unadjusted ATE estimator using the simple difference in the outcomes (Unad).

The  MB method was described in Section 4.  We also consider  kernel-based Mahalanobis balancing (kernelMB),  where we use $\Phi(X)=(K(X,X_1),...,K(X,X_n))^\top$ with $K(\cdot,\cdot)$ set to be the Gaussian kernel. High-dimensional Mahalanobis balancing  (hdMB) was described in Section A, where $\Phi(X)=(X_{(1)},\cdots, X_{(K_0)})^\top$ and $K_0$ is selected by
	Algorithm S1. We compare them with four  high-dimensional regularized balancing methods (RCAL1, RCAL2, ARB, hdCBPS). We do not report imbalance measures for   ARB and hdCBPS, because they  utilize outcome information in weight construction. For all methods except kernelMB, we set $\Phi(X)=X$ and thus $K=p$. In each scenario, we consider $(n,p)=(200,100)$ or $(200,500)$.    Table \ref{Table 2} summarizes the outputs.

	In Scenario F,  we generate the covariates from $X \sim N({0}, {\Sigma})$ where $Cov(X_j,X_k)=2^{-I(j \neq k)}$.   The treatment assignment is generated by $T  \sim Bernoulli(\pi(X))$ with $ \pi(X)=1/\{1+\exp({X_{1}}+\sum_{j=2}^{6}X_{j}/2 )\}$. Therefore, treatment assignment is correlated with a sparse subset of covariates. The outcome  is generated from $Y=T(\sum_{j=1}^{5}X_{j})+(1-T)(\sum_{j=1}^{5}X_{j}/2) +\epsilon$.
	
	Both MB and kernelMB have low bias, small standard deviation, and small imbalance measures when $(n,p)=(200,100)$, implying that they are capable of approximately balancing  covariates  when the covariate dimension is substantially smaller than the sample size. When  $(n,p)=(200,500)$, MB and kernelMB exhibit quite large GMIM values, and  their biases are  larger than those of the high-dimensional regularized balancing methods but are smaller than that of MDABW. The hdMB method    shrinks the GMIM  for the selected covariates to be zero, but its GMIM  value is very large when all covariates are used to calculate this measure. It  achieves lower bias and RMSE compared to  MDABW and   all   high-dimensional regularized balancing methods (RCAL1, RCAL2, ARB, hdCBPS).  In conclusion, the hdMB method is recommended over MB and kernelMB when the true propensity score model is sparse. It is competitive to the existing  high-dimensional regularized balancing methods.
	
	In Scenario G,  the data generation procedure of the covariates  is the same as the one in Scenario F.  The treatment index is generated by $T  \sim Bernoulli(\pi(X))$ with $ \pi(X)= 1/\{1+\exp(X_{1}+\sum_{j=2}^{5}X_{j}/{2}+ 10 \sum_{j=6}^{p}X_{j}/p) \}$. Note that treatment assignment is correlated with all covariates, and thus the propensity score model is dense. The outcome  is generated from $Y=T(10\sum_{j=1}^{p}X_{j}/p)+(1-T)(5\sum_{j=1}^{p}X_{j}/p) +\epsilon$.

	Both MB and kernelMB have lower bias and smaller standard deviation compared to hdMB and MDABW.   The   high-dimensional regularized balancing methods (RCAL1, RCAL2, ARB, hdCBPS) have  larger biases than MB and kernelMB.  Both  MB and kernelMB produce much smaller GMIM values   than other methods when $(n,p)=(200,100)$. The  MB and kernelMB methods are  recommended since they enjoy lowest biases, smallest RMSEs, and smallest GMIMs. The hdMB method is not recommended in this dense scenario.
	
	\subsection{Additional Simulation for Model Misspecification}
	
	We consider high-dimensional situation with $(n,p)=(200,100)$ when the outcome model is misspecified as a linear model. The data generation procedure of the covariates  is the same as the one in Scenario F.
	
	In the Scenario M1, we consider a sparse propensity score model $$\pi(X_{i})= \frac{1}{1+\exp (X_{i 1}+\sum_{j=2}^{6}X_{i j}/{2})},$$ and  a nonlinear outcome model $$Y_{i}=2T_{i}\left(\sum_{j=1}^{6}X_{ij} + \sum_{j=1}^{6}X^2_{ij}\right)+(1-T_{i})\left(\sum_{j=1}^{6}X_{ij} + \sum_{j=1}^{6}X^2_{ij}\right)+\epsilon_{i},$$
	where $\epsilon_{i}$ is  a standard normal random variable.
	
	In the Scenario M2, we consider a dense propensity score model
	$$\pi(X_{i})=\frac{1}{1+\exp (X_{i 1}+\sum_{j=2}^{5}X_{i j}/{2}+  \sum_{j=6}^{100}{ X_{i j}}/10)},$$ and  a nonlinear outcome model $$Y_{i}=T_{i}\left(\sum_{j=1}^{100}X_{ij} + \sum_{j=1}^{50}X^2_{ij}\right)/10
	+(1-T_{i})\left(\sum_{j=1}^{100}X_{ij} + \sum_{j=1}^{50}X^2_{ij}\right)/20+\epsilon_{i}.$$

	Table \ref{TablesMis} gives the results. We observe that MB and hdMB have best performance in Scenario M1, and kernelMB has best performance in Scenario M2.  Mahalanobis balancing   does not use information of the outcomes, and thus is  robust in these scenarios.
	
	\begin{table}[t]
		\caption{Model misspecification scenarios.}
		\label{TablesMis}
		\resizebox{\textwidth}{!}{
			\begin{tabular}{cccccccccccc}
				\toprule
				\textbf{Scenario M1}& \multicolumn{1}{c}{{\textbf{Bias}}} & \multicolumn{1}{c}{\textbf{SD}}& \multicolumn{1}{c}{\textbf{RMSE}} & \multicolumn{1}{c}{\textbf{maxASMD}} & \multicolumn{1}{c}{\textbf{meanASMD}} & \multicolumn{1}{c}{\textbf{medianASMD}} &\multicolumn{1}{c}{\textbf{GMIM}} \\
				\midrule
				\textbf{Unad} &-3.45&0.50&3.49&1.54&0.16&0.10&3.22\\
				{\textbf{MB}}  &-1.47&0.61&1.59&0.58&0.06&0.05&0.52\\
				{\textbf{kernelMB}} &-1.52&0.56&1.62&0.80&0.08&0.06&0.71\\
				{\textbf{hdMB}}
				&-0.72&0.73&1.03&1.25&0.14&0.10&4.19\\
				\textbf{MDABW}  &-2.50&0.96&2.68&1.54&0.14&0.11&2.12\\
				\textbf{RCAL1}    &-1.68&0.60&1.79&1.10&0.15&0.11&2.17\\
				\textbf{RCAL2}
				&-1.67&0.60&1.78&1.49&0.15&0.11&2.76
				\\
				\textbf{ARB}  &-1.22&0.55&1.34&-&-&-&-\\
				\textbf{hdCBPS}  &-1.68&0.61&1.80&-&-&-&-\\
				
				\midrule
				\textbf{Scenario M2}& \multicolumn{1}{c}{{\textbf{Bias}}} & \multicolumn{1}{c}{\textbf{SD}} & \multicolumn{1}{c}{\textbf{RMSE}} & \multicolumn{1}{c}{\textbf{maxASMD}} & \multicolumn{1}{c}{\textbf{meanASMD}} & \multicolumn{1}{c}{\textbf{medianASMD}} &\multicolumn{1}{c}{\textbf{GMIM}}\\
				\midrule
				\textbf{Unad} &-1.35&0.36&1.40&1.30&0.20&0.17&3.53\\
				{\textbf{MB}} &-0.50&0.48&0.68&0.56&0.07&0.06&0.64\\
				{\textbf{kernelMB}} &-0.35&0.37&0.51&0.74&0.08&0.07&0.82\\
				{\textbf{hdMB}}
				&-0.70&0.69&0.98&1.28&0.16&0.11&4.87\\
				\textbf{MDABW} &-1.02&0.51&1.14&1.30&0.15&0.14&1.86\\
				\textbf{RCAL1}   &-1.20&0.41&1.27&0.93&0.20&0.18&3.03\\
				\textbf{RCAL2}
				&-1.19&0.41&1.26&1.24&0.20&0.18&3.18
				\\
				\textbf{ARB}  &-0.78&0.45&0.90&-&-&-&-\\
				\textbf{hdCBPS}  &-1.21&0.42&1.28&-&-&-&-\\
				\bottomrule
			\end{tabular}
		}
	\end{table}
	
	\clearpage
	
	\section{Lemma and Proof}
	
	Lemma 1 is similar to Lemma 8 by  \cite{athey2016approximate}, but the assumptions are  different.
	\begin{lemma}(Weight Behaviour I)
		Suppose that $\sum_{\left\{i: T_{i}=1\right\}} w_{i}^{3} / \|w\|_{2}^{3} = o_p(1)$ and that the $\varepsilon_{1i}$ are sub-Gaussian. Then, as $n\rightarrow \infty$,
		\begin{equation*}
			\frac{1}{\|w\|_{2}} \sum_{\left\{i: T_{i}=1\right\}} w_{i} \varepsilon_{1i} \stackrel{d}{\rightarrow} {N}\left(0, \sigma^{2}\right),
			\label{noise-beha}
		\end{equation*}
		where $\sigma^2=Var(\varepsilon_{1i})$.
	\end{lemma}
	
	\begin{proof}
		Since the MB weights do not utilize the outcomes, the $\varepsilon_{1i}$ are independent of the $w_{i}$ given the $X_{i}$. Therefore,
		$$
		E\{\sum_{\{i: T_{i}=1\}} w_{i} \varepsilon_{1i} \mid w_{i}, T_i=1\}=0 \text { and } \operatorname{Var}\{\sum_{\{i: T_{i}=1\}} w_{i} \varepsilon_{1i} \mid w_{i}, T_i=1\}=\sigma^{2}\|w\|_{2}^{2}.
		$$
		Since the  $\varepsilon_{1i}$ are sub-Gaussian,
		$$
		E\{\sum_{\{i: T_{i}=1\}}(w_{i} \varepsilon_{1i})^{3} \mid w_{i}, T_i=1\}\leq C \sum_{\{i: T_{i}=1\}} w_{i}^{3}
		$$
		for some positive constant $C$. Therefore,
		$$
		\frac{E\{\sum_{\{i: T_{i}=1\}}(w_{i} \varepsilon_{1i})^{3} \mid w_{i}, T_i=1\}}
		{{Var}\{\sum_{\left\{i: T_{i}=1\right\}} w_{i} \varepsilon_{1i} \mid  w_{i}, T_i=1\}^{3 / 2}}
		=O\left(\sum_{\left\{i: T_{i}=1\right\}} w_{i}^{3} / \|w\|_{2}^{3}\right)=o_{p}(1).
		$$
		Hence,   the Lyapunov condition is verified, and the proof is completed by applying   Lyapunov's central limit theorem.
	\end{proof}
	
	The next lemma asserts  that the MB weights  are stable  when the entropy loss function is employed.
	\begin{lemma}(Weight Behaviour II)
		Suppose that  the entropy loss function is used. Assume that $E\{\exp(a\theta^{\top}W^{1 / 2}[\Phi(X)-E\{\Phi(X) \mid T=1 \} ])\mid T=1\} \leq K_a$ for all $\theta \in \Theta$, where $K_a$ is some positive constant, $a = 1,2,3$. Assume that  $\theta^{\top}W^{1 / 2}\left(\Phi\left(X\right)-E\{\Phi\left(X \right)\mid T=1\} \right)$ is sub-Gaussian. Then   the following properties  hold for the MB weights:
		\begin{eqnarray*}
			\sum_{i:T_{i}=1}(\hat{w_i}^{MB})^a&=& O_p(n^{1-a}),\ \ \text{ for } a = 1,2,3;\\
			\max_{i:T_{i}=1}\hat{w_i}^{MB}&=& O_p(n^{-1/2}).
		\end{eqnarray*}
	\end{lemma}
	
	\begin{proof}
		When the entropy loss function is employed,
		\begin{eqnarray*}
			\hat{w_i}^{MB}  &=& \frac{\exp\left\{\hat{\theta}^\top {W}^{1/2}\Phi(X_i)\right\}}{\sum_{j:T_j=1}\exp\left\{\hat{\theta}^\top {W}^{1/2}\Phi(X_j)\right\}}\\
			&=&  \frac{\exp\left\{\hat{\theta}^{\top} W^{1 / 2}\Phi\left(X_{i}\right)\right\}}{n_1E\left[\exp\left\{\hat{\theta}^{\top} W^{1 / 2}\Phi\left(X\right)\right\}\mid T=1 \right]}(1 + o_p(1))\\
			&\leq& \frac{\exp\left\{\hat{\theta}^{\top} W^{1 / 2}\Phi\left(X_{i}\right)\right\}}
			{n_1\exp\left(\hat{\theta}^{\top} W^{1 / 2}E\left\{\Phi\left(X\right)\mid T=1 \right\}\right)}(1 + o_p(1))\\
			& =&\frac{1}{n_1}\exp\left(\hat{\theta}^{\top} W^{1 / 2}\left[\Phi\left(X_{i}\right)-E\left\{\Phi\left(X\right)\mid T=1 \right\}\right]\right)(1 + o_p(1))\\
			&=& O_p(n^{-1}).
		\end{eqnarray*}
		Therefore,
		\begin{eqnarray*}
			\sum_{i:T_{i}=1}(\hat{w_i}^{MB})^a &=&
			\frac{1}{{n_1}^{a}}\sum_{i: T_{i}=1} \exp\left(a\hat{\theta}^{\top} W^{1 / 2}\left[\Phi\left(X_{i}\right)-E\left\{\Phi\left(X\right)\mid T=1 \right\}\right]\right)(1 + o_p(1))\\
			& =& n_1^{1-a} E\{\exp(a\theta^{\top}W^{1 / 2}[\Phi(X)-E\{\Phi(X) \mid T=1 \} ])\mid T=1\} (1 + o_p(1))\\
			& = &n_1^{1-a}   K_{a} (1 + o_p(1))\\
			&= &O_p(n^{1-a}).
		\end{eqnarray*}
		
		Moreover,  by the maximal inequality  \citep[e.g.,][]{rigollet2015high}, we  obtain that
		$$Pr\left(\max_{i:T_i=1}\frac{1}{n_1}\exp\left(\hat{\theta}^{\top} W^{1 / 2}(\Phi\left(X_{i}\right)-E\{\Phi\left(X\right)\mid T=1\})\right)\geq \frac{1}{n_1}\exp(t)\right) \leq n_1  \exp\left(-\frac{t^2}{2\sigma^2}\right).$$
		Set $t=\sqrt{2}\sigma\log(n_1^{1/2}s)$, we have  $\max_{i:T_i=1}\hat{w_i}^{MB} = O_p(n^{-1/2})$.
	\end{proof}

	The following lemma shows the asymptotic property for the solution $\hat{\theta}$ of the MB dual problem.
	\begin{lemma}
		Assume  that  $1/\pi(X)= C*\tilde{f}^{*'}(\theta^{*\top} {W}^{1/2}(\Phi(X)-\bar{\Phi}))$ for some $\theta^{*} \in \operatorname{int}(\Theta)$ and some constant $C$ for all $X$. Suppose that Assumption 3 holds. Then ${\hat{\theta}}$ is consistent for $\theta^*$. Moreover,   ${\hat{\theta}}$  is asymptotically normal.
	\end{lemma}
	
	\begin{proof}
		The first order optimality condition for the dual  problem is
		$$
		\begin{aligned}
			\sum_{i=1}^{n}T_i\tilde{f}^{*'}\left(\hat{\theta}^{\top} W^{1 / 2}\left(\Phi\left(X_{i}\right)-\bar{\Phi}\right)\right)(\Phi_j\left(X_{i}\right)
			-\bar{\Phi}_j) + \sqrt{\delta}\frac{\hat{\theta}_j}{||\hat{\theta}||_2}=0, \text{ \quad $j=1,\cdots,K$, }
		\end{aligned}
		$$
		where $\Phi_j\left(X_{i}\right)$, $\bar{\Phi}_j$, and $\hat{\theta}_j$ are the $j$th components of $\Phi\left(X_{i}\right)$, $\bar{\Phi}$, and $\hat{\theta}$, respectively.
		Write
		$$
		\begin{aligned}
			\frac{1}{n}\sum_{i=1}^{n} \psi(X_i,T_i;{\theta})=		\frac{1}{n}\sum_{i=1}^{n}T_i\tilde{f}^{*'}\left({\theta}^{\top} W^{1 / 2}\left(\Phi\left(X_{i}\right)-E\{\Phi\left(X\right)\}\right)\right)
			(\Phi\left(X_{i}\right)
			-E\{\Phi\left(X\right)\}),
		\end{aligned}
		$$
		which is a set of $K$ estimating functions.
		Note that
		$$	
		\begin{aligned}
			E\left\{ \psi(X_i,T_i;{\theta})\right\}&=E\left\{E\left( \psi(X,T;{\theta})\right)\mid X\right\}\\
			&= E\left[\pi(X)\tilde{f}^{*'}\left({\theta}^{\top} W^{1 / 2}\left(\Phi\left(X\right)-E\{\Phi\left(X\right)\}\right)\right)W^{1 / 2}(\Phi\left(X\right)-E\{\Phi\left(X\right)\})\right].
		\end{aligned}
		$$
		For the above conditional expectation to be zero, it must be true that $$\pi(X)\tilde{f}^{*'}\left({\theta}^{\top} W^{1 / 2}\left(\Phi\left(X\right)-E\{\Phi\left(X\right)\}\right)\right)$$ is a constant for any $X$. By the assumption that
		$1/\pi(X)= C*\tilde{f}^{*'}(\theta^{*\top} {W}^{1/2}(\Phi(X)-\bar{\Phi}))$, we obtain that $\theta^{*}$ is the unique solution of $E\left\{ \psi(X_i,T_i;{\theta})\right\}=0$. Therefore, by the estimating equation theory \citep{van2000asymptotic}, the solution of the estimating equations
		$$	\frac{1}{n}\sum_{i=1}^{n} \psi(X_i,T_i;{\theta})=0,$$
		denoted by $\tilde{\theta}$, is asymptotically consistent for $\theta^*$. Moreover, by the assumption that  $\delta=o(n)$, we obtain $\frac{1}{n}\sqrt{\delta}\frac{{\theta}_j}{||{\theta}||_2}=o_p(n^{-1/2})$ for any $\theta \in \operatorname{int}(\Theta)$. Therefore, the difference between $\hat{\theta}$ and $\tilde{\theta}$ is asymptotically negligible, and  thus $\hat{\theta}$ is asymptotically consistent for $\theta^*$. Furthermore, by Taylor expansion, we obtain that as $n\rightarrow \infty$,
		$$\sqrt{n}(\hat{\theta}-\theta^*)\stackrel{d}{\longrightarrow} N(0, \Sigma),$$
		where $\Sigma=\left\{E(\frac{\partial{\psi}}{\partial \theta^\top})\right\}^{-1}E(\psi \psi^\top)\left\{E(\frac{\partial{\psi}}{\partial \theta})\right\}^{-1}$.
	\end{proof}
	
	The following lemma asserts that the MB weight is close to the unknown inverse probability weight. The proof is similar to  arguments by \cite{lee2023improving}.
	\begin{lemma}(Weight Behaviour III)
		Assume  that  $1/\pi(X)= C*\tilde{f}^{*'}(\theta^{*\top} {W}^{1/2}(\Phi(X)-\bar{\Phi}))$ for some $\theta^{*} \in \operatorname{int}(\Theta)$ and some constant $C$ for all $X$. Suppose that Assumption 3 holds. Then the following property holds:
		\begin{equation*}
			n \hat{w_i}^{MB}=\frac{1}{\pi(X_i)}+O_p(n^{-\frac{1}{2}}).
		\end{equation*}
	\end{lemma}
	
	\begin{proof}
		$$
		\begin{aligned}
			\frac{1}{n}\sum_{i=1}^{n} T_i\tilde{f}^{*'}\left(\hat{\theta}^{\top} W^{1 / 2}\left(\Phi\left(X_{i}\right)-\bar{\Phi}\right)\right)
			&=\frac{1}{n}\sum_{i=1}^{n} T_i\tilde{f}^{*'}\left(\hat{\theta}^{\top} W^{1 / 2}\left(\Phi\left(X_{i}\right)-E\{\Phi\left(X\right)\}\right)\right) + O_p(n^{-\frac{1}{2}})\\
			&=\frac{1}{n}\sum_{i=1}^{n} T_i\tilde{f}^{*'}\left({\theta}^{*\top} W^{1 / 2}\left(\Phi\left(X_{i}\right)-E\{\Phi\left(X\right)\}\right)\right)  + O_p(n^{-\frac{1}{2}}) \\
			&=\frac{1}{n}\sum_{i=1}^{n}\frac{T_i}{C\pi(X_i)}+O_p(n^{-\frac{1}{2}})\\
			&= \frac{1}{C}+o_p(1).
		\end{aligned}
		$$
		Therefore,
		$$
		\begin{aligned}
			n\hat{w_i}^{MB}&=\frac{\tilde{f}^{*'}\left(\hat{\theta}^{\top} W^{1 / 2}\left(\Phi\left(X_{i}\right)-\bar{\Phi}\right)\right)}{\frac{1}{n}\sum_{i: T_i=1}\tilde{f}^{*'}\left(\hat{\theta}^{\top} W^{1 / 2}\left(\Phi\left(X_{i}\right)-\bar{\Phi}\right)\right)}\\
			&=\frac{\tilde{f}^{*'}\left({\theta}^{*\top} W^{1 / 2}\left(\Phi\left(X_{i}\right)-E\{\Phi\left(X\right)\}\right)\right)+O_p(n^{-1/2})}
			{1/C+o_p(1)}\\
			&=\frac{1}{\pi(X_i)}+O_p(n^{-\frac{1}{2}}).
		\end{aligned}
		$$
	\end{proof}
	
	\subsection{Proof for Theorem 1}
	
	\begin{proof}
		We  utilize the Fenchel duality theorem  \citep[Theorem B.39]{mohri2018foundations}. Without loss of generability, suppose that subjects $i=1,\cdots, n_1$ are in the treated group, and subjects $i=n_1+1,\cdots, n$ are in the control group. The primal problem (7) is
		\begin{equation*}
			\underset{w \in \mathbb{R}^{n_1}}{\text{minimize}}: \ \sum_{i=1}^{n_1}\tilde{f}(w_i)+I_{\mathcal{C}}\left(\sum_{i=1}^{n_1}w_i{W}^{1/2}(\Phi(X_i)-\bar{\Phi})\right),
			\label{PrimalS}
		\end{equation*}
		where $w=(w_1,\cdots, w_{n_1})^\top$.  Let $F(w)=\sum_{i=1}^{n_1}\tilde{f}(w_i)$. The conjugate function of $F$ is given by
		$$
		\begin{aligned}
			F^*(w)&=\sup_{v}\left(\sum_{i=1}^{n_1} w_i v_i-F(v)\right)\\
			&=\sup_{v}\sum_{i=1}^{n_1} \left(w_i v_i-\tilde{f}(v_i)\right)\\
			&=\sum_{i=1}^{n_1}\sup_{v_i} (w_i v_i-\tilde{f}(v_i))\\
			&=\sum_{i=1}^{n_1}\tilde{f}^*(w_i).
		\end{aligned}
		$$
		Note that if the normalization constraint is added to the primal problem, then the conjugate function $F^*$ is no longer additive.
		Let $g(\theta)=I_{\mathcal{C}}(\theta)$, for any $\theta\in \mathbb{R}^K$.   The conjugate function of $g$ is given by
		$$
		\begin{aligned}
			g^*(\theta)&=\sup_{u}\left(\sum_{k=1}^K \theta_k u_k-I_{\mathcal{C}}(u)\right)\\
			&=\sup_{\|u\|_2\leq \sqrt{\delta}}\left(\sum_{k=1}^K \theta_k u_k\right)\\
			&=\sup_{\|u\|_2\leq \sqrt{\delta}}\left(\|\theta\|_2\|u\|_2\right)\\
			&=\sqrt{\delta}\|\theta\|_2.
		\end{aligned}
		$$
		Define the mapping $A:\mathbb{R}^{n_1}\rightarrow\mathbb{R}^K$ such that $Aw=\sum_{i=1}^{n_1}w_i{W}^{1/2}(\Phi(X_i)-\bar{\Phi})$. Then $A$ is a bounded linear map. Let $A^*$ be the adjoint operator of $A$. Then for all $\theta=(\theta_1,\cdots, \theta_K)^\top\in \mathbb{R}^{K}$,
		$$
		\begin{aligned}
			A^*\theta&=(\sum_{k=1}^K \theta_k{W}^{1/2}(\Phi_k(X_1)-\bar{\Phi}_k),\cdots, \sum_{k=1}^K \theta_k{W}^{1/2}(\Phi_k(X_{n_1})-\bar{\Phi}_k))^\top\\
			&=(\theta^\top {W}^{1/2}(\Phi(X_1)-\bar{\Phi}),\cdots, \theta^\top {W}^{1/2}(\Phi(X_{n_1})-\bar{\Phi}))^\top.
		\end{aligned}
		$$
		
		Define $\theta_0=Aw_0=\sum_{i=1}^{n_1} {W}^{1/2}(\Phi(X_i)-\bar{\Phi})/n_1^a$, where $w_0=(1/n_1^a,\cdots, 1/n_1^a)^\top$. Here, we choose $a$ to be sufficiently large such that $\|\theta_0\|_2< \sqrt{\delta}$. Since each component of $w_0$ is non-negative, $w_0\in dom({F})$, where $dom(F)$ denotes the domain of $F$. Therefore,  $\theta_0\in A(dom(F))$.  Since $\|\theta_0\|_2< \sqrt{\delta}$, we obtain that $g(\theta_0)=0$   and $g$ is continuous at $\theta_0$. Therefore, $\theta_0\in A(dom(F))\cap cont(g)$, implying that $ A(dom(F))\cap cont(g)\neq \emptyset$, where $cont(g)$ is the set of continuous points of $g$. Therefore, the strong duality condition of the  Fenchel duality theorem is verified.
		Moreover,
		$$
		\begin{aligned}
			F(A^*\theta)+g^*(-\theta)&=\sum_{i=1}^{n_1}\tilde{f}^*(\theta^\top {W}^{1/2}(\Phi(X_i)-\bar{\Phi}))+\sqrt{\delta}\|\theta\|_2.
		\end{aligned}
		$$
		The  Fenchel duality theorem leads to that
		$$
		\begin{aligned}
			&\underset{w \in \mathbb{R}^{n_1}}{\text{minimize}}: \ \sum_{i=1}^{n_1}\tilde{f}(w_i)+I_{\mathcal{C}}\left(\sum_{i=1}^{n_1}w_i{W}^{1/2}(\Phi(X_i)-\bar{\Phi})\right)\\
			=&\underset{\theta \in \mathbb{R}^K}{\text{minimize}}: \ \sum_{i=1}^{n_1}\tilde{f}^*\left(\theta^\top {W}^{1/2}(\Phi(X_i)-\bar{\Phi})\right)+\sqrt{\delta}\lVert \theta\rVert_2.
		\end{aligned}
		$$
		Furthermore, since the strong duality condition holds,  equality of the Fenchel's inequality \citep[Prop. 38]{mohri2018foundations} holds. leading to that
		$A^*\hat{\theta}$ is a subgradient of $F$ at $\hat{w}$. We consider the simpler situation that the loss function $f(\cdot)$ is differentiable and it has non-negative domain as in \citep{Josey2021}. Then, $A^*\hat{\theta}=F^\prime(\hat{w})$, or equivalently,
		$$\hat{\theta}^\top {W}^{1/2}(\Phi(X_i)-\bar{\Phi})
		=\tilde{f}^\prime(\hat{w}_i).
		$$
	 Therefore,  $\hat{w}_i = (\tilde{f}^{\prime})^{-1}(\hat{\theta}^{\top}W^{1 / 2}\left(\Phi\left(X_{i}\right)-\bar{\Phi}\right))=(\tilde{f}^{*})^{\prime}(\hat{\theta}^{\top}W^{1 / 2}\left(\Phi\left(X_{i}\right)-\bar{\Phi}\right))$, $i=1,\cdots, n_1$,  where the second equality holds by Fenchel's identity \citep{RyuYin2022,MordukhovichNam2022}.

	\end{proof}
	
	\subsection{Proof for Theorem 2}
	
	\begin{proof}
		We only need to show that $	\widehat{\tau}_1^{MB} =\tau_1+ O_p(n^{-1/2})$, where $\widehat{\tau}_1^{MB}  = \sum_{i: T_i=1} \hat{w_i}^{MB}Y_i$.
		First,  assume the propensity score satisfies that $1/\pi(X)= C*\tilde{f}^{*'}(\theta^{*\top} {W}^{1/2}(\Phi(X)-\bar{\Phi}))$ for some $\theta^{*} \in \operatorname{int}(\Theta)$ and some constant $C$ . By Lemma 3,
		$$
		\begin{array}{ll}
			\sum_{i: T_i=1} \hat{w_i}^{MB}Y_i &= \sum_{i: T_i=1}\left(\hat{w_i}^{MB} - \frac{1}{n\pi(X_i)}\right)Y_i + \sum_{i: T_i=1}\frac{Y_i}{n\pi(X_i)} \\
			& =  \sum_{i: T_i=1} O_p(n^{-3/2})O_p(1) + O_p(n^{-1/2})  + \tau_1.
		\end{array}
		$$
		Therefore,  $\widehat{\tau}_1-\tau_1 = O_p(n^{-1/2})$.
		
		Second,
		we assume that the conditional potential outcome $E\{Y(1)|X\}=\beta^{\top} W^{1 / 2} \Phi(X)$ with parameter $\beta$,    $\sum_{i: T_i=1} (\hat{w_i}^{MB})^{3} / \|\hat{w}^{MB}\|_2^3 = o_p(1)$, and $\|\hat{w}^{MB}\|_2^2 = O_p(n^{-1})$.
		Note that the following decomposition holds:
		$$
		\begin{array}{ll}
			\sum_{i:T_i=1} \hat{w_i}^{MB}Y_i -\tau_1
			=& \beta^{\top} \sum_{i:T_i=1}\hat{w_i}^{MB} W^{1 / 2}(\Phi(X_i) - \bar{\Phi}) + \sum_{i:T_i=1}\hat{w_i}^{MB}\epsilon_{1i} \\
			& + \frac{1}{n}\sum_{i=1}^{n}\beta^{\top}W^{1 / 2}\Phi(X_i) -\tau_1.
		\end{array}
		$$
		Let the $\hat{w_i}$  be the unnormalized MB weights. By Assumption 3.4,
		$$
		\begin{array}{ll}
			\sum_{i:T_i=1}\hat{w_i} = \sum_{i:T_i=1}\exp\left( \theta^{\top}W^{1/2}\left(\Phi(X_i) - \bar{\Phi}\right) \right)=O_p(n).
		\end{array}
		$$
		Therefore,
		$$
		\begin{array}{ll}
			|\beta^{\top} \sum_{i:T_i=1}\hat{w_i}^{MB} W^{1 / 2}(\Phi(X_i) - \bar{\Phi})|
			&\leq \|\beta\|_2 \|\sum_{i:T_i=1}\hat{w_i}^{MB} W^{1 / 2}(\Phi(X_i) - \bar{\Phi})\|_2 \\
			&=  \|\beta\|_2 \|\sum_{i:T_i=1}\hat{w_i}W^{1 / 2}(\Phi(X_i) - \bar{\Phi})\|_2 / \sum_{i:T_i=1}\hat{w_i}\\
			&\leq \|\beta\|_2 \sqrt{\delta} / \sum_{i:T_i=1}\hat{w_i}\\
			&= o_p(n^{-1/2}).
		\end{array}
		$$
		By Lemma 1, we have $\frac{1}{\|\hat{w}^{MB}\|_2}\sum_{i:T_i=1}\hat{w_i}^{MB}\epsilon_{1i}=O_p(1)$. Therefore, by the assumption that   $\|\hat{w}^{MB}\|_2^2 = O_p(n^{-1})$, we obtain
		$
		\sum_{i:T_i=1}\hat{w_i}^{MB}\epsilon_{1i} = O_p(\|\hat{w}^{MB}\|_2) = O_p(n^{-1/2})
		$.
		Finally,\\
		$
		\frac{1}{n}\sum_{i=1}^{n}\beta^{\top}W^{1 / 2}\Phi(X_i) -\tau_1 = O_p(n^{-1/2}).
		$
		Therefore,  $\widehat{\tau}_1-\tau_1 = O_p(n^{-1/2})$. The double robustness property is proved.
	\end{proof}

	\subsection{Proof for Theorem 3}
	\begin{proof}

    By the assumption for $\theta^{\top}W^{1/2}\left(E\{\Phi(X_i)\} - E\{\Phi(X_i)\mid T=1\}\right)$, we obtain \\ $\exp\left(\theta^{\top}W^{1/2}\left(E\{\Phi(X_i)\} - E\{\Phi(X_i)\mid T=1\}\right)\right) = O_p(n^{\gamma})$. It follows that
	$$
	\begin{array}{ll}
		\|\sum_{i:T_i=1}\hat{w}_i^{MB} W^{1 / 2}(\Phi(X_i) - \bar{\Phi})\|_2
		&\leq \sqrt{\delta} O_p(n^{-1})\\
		&\ \ \times \exp(\theta^{\top}W^{1/2}\left(E\{\Phi(X_i)\} - E\{\Phi(X_i)\mid T=1\} )+o_p(1)\right)\\
		&=O_p(n^{s-1 + \gamma}).
	\end{array}
	$$
	Therefore,
	$$
	\begin{array}{cc}
		& |\beta^{\top} \sum_{i:T_i=1}\hat{w_i}^{MB} W^{1 / 2}(\Phi(X_i) - \bar{\Phi})|\leq  \|\beta\|_2 \|\sum_{i:T_i=1}\hat{w_i}^{MB} W^{1 / 2}(\Phi(X_i) - \bar{\Phi})\|_2 \leq  O_p(n^{s+ \gamma + \alpha}).
	\end{array}
	$$
	
	By Lemmas 1 and 2,
	$
	\sum_{i:T_i=1}\hat{w}_i^{MB}\epsilon_{1i} = O_p(\|\hat{w}^{MB}\|_2) = O_p(n^{-1/2})
	$.
	Here, the condition $\sum_{\left\{i: T_{i}=1\right\}} w_{i}^{3} / \|w\|_{2}^{3} = o_p(1)$ holds. Moreover,
	$$
	\frac{1}{n}\sum_{i=1}^{n}\beta^{\top}W^{1/2}\Phi(X_i) - \tau_1 = \frac{1}{n}\sum_{i=1}^{n}\mu_1(X_i) -\tau_1  = O_p(n^{-1/2}).
	$$
	Therefore,
	$$
	\begin{array}{ll}
		\widehat{\tau_1}^{MB}-\tau_1  &=  \beta^{\top} \sum_{i:T_i=1}\hat{w}_i^{MB} W^{1 / 2}(\Phi(X_i) - \bar{\Phi}) + \sum_{i: T_i=1}\hat{w}_i^{MB}\epsilon_{1i} \\
		&\ \ + \frac{1}{n}\sum_{i=1}^{n}\beta^{\top}W^{1 / 2}\Phi(X_i) - \tau_1\\
		&= O_p(n^{-1/2}) + O_p(n^{s+ \gamma + \alpha}).
	\end{array}
	$$
	
	\end{proof}

 	\subsection{Proof of  Theorem 4}
	\begin{proof}
		To simplify the notations, let  $\pi_i=\pi(X_i)$. Write $E\{Y(1)|X\}=\mu_1(X)=\beta^{\top} W^{1 / 2} \Phi(X)$ with parameter $\beta$. Then the following decomposition holds:
		$$
		\begin{aligned}
			\hat{\tau_1}-\tau_1&=\sum_{i=1}^{n} T_{i} \hat{w}^{MB}_{i} Y_{i}-\tau_1\\
			&=\sum_{i=1}^{n}T_{i}\left(\hat{w}^{MB}_{i} -\frac{1}{n\pi_i}\right)\left(Y_i-\mu_1\left(X_{i}\right)\right) +\sum_{i=1}^{n}\frac{T_{i}}{n\pi_i}\left(Y_i-\mu_1\left(X_{i}\right)\right)\\
			&+\sum_{i=1}^{n}\left(T_{i}\hat{w}^{MB}_{i} -\frac{1}{n}\right)\beta^{\top} W^{1 / 2} \Phi\left(X_{i}\right)+\left(\frac{1}{n}\sum_{i=1}^{n}\mu_1(X_i)-\tau_1\right)\\
			&=A_{1}+R_1+A_{2}+R_2.
		\end{aligned}
		$$

        Firstly, following the proof of theorem 2 in \cite{wang2020minimal},  we show that
        $$
        \sup_{x \in \mathcal{X}}|\hat{w}^{MB}_{i} -\frac{1}{n\pi_i}| = o_p(1).
        $$
        \begin{lemma}

        There exists a global minimizer $\hat{\theta}$ such that
        $$
        \left\|\hat{\theta}-\theta_{true}\right\|_2=O_p\left(n^{-\frac{1}{4} + \epsilon }\right) \text{ for arbitrary $\epsilon > 0$}
        $$
        for the objective function $G(\cdot)$:
        $$
        G(\theta):= \sum_{i=1}^{n}T_i\tilde{f}^{*}\left({\theta}^{\top} W^{1 / 2}\left(\Phi\left(X_{i}\right)-\bar{\Phi}\right)\right)+\delta^{\frac{1}{2}}||\theta||_2 ,
        $$
        \end{lemma}

        To show that a minimizer $\Delta^*$ of $G\left(\theta_{true}+\Delta\right)$ exists in
        $$
        \mathcal{C}=\left\{\Delta \in \mathcal{R}^K:\|\Delta\|_2 \leq C n^{s - 1 + \epsilon}\right\}
        $$
        for some constant $C$, it suffices to show that
        \begin{equation}
        E\left\{\inf _{\Delta \in \mathcal{C}} G\left(\theta_{true}+\Delta\right)-G\left(\theta_{true}\right)>0\right\} \rightarrow 1, \text { as } n \rightarrow \infty
        \label{eq:verify}
        \end{equation}
        by the continuity of $G(\cdot)$.

        To show \eqref{eq:verify}, we use mean value theorem: for some $\tilde{\theta}$ between $\theta_{true}$ and $\hat{\theta}$,
        $$
        \begin{aligned}
        & G\left(\theta_{true}+\Delta\right)-G\left(\theta_{true}\right) \\
        = & \sum_{i=1}^{n}T_i\left(\tilde{f}^{*}\left((\theta_{true}+\Delta)^{\top} W^{1 / 2}\left(\Phi\left(X_{i}\right)-\bar{\Phi}\right)\right) - \tilde{f}^{*}\left({\theta_{true}}^{\top} W^{1 / 2}\left(\Phi\left(X_{i}\right)-\bar{\Phi}\right)\right)\right) \\
        + & \delta^{\frac{1}{2}}||\theta_{true}+\Delta||_2 - \delta^{\frac{1}{2}}||\theta_{true}||_2 \quad (||\theta_{true}+\Delta||_2 - ||\theta_{true}||_2 \geq - ||\Delta||_2.) \\
        \geq & \Delta \cdot \sum_{i=1}^{n}T_i\tilde{f}^{*'}\left((\theta_{true}+\Delta)^{\top} W^{1 / 2}\left(\Phi\left(X_{i}\right)-\bar{\Phi}\right)\right)\left(\Phi\left(X_{i}\right)-\bar{\Phi}\right) \\
        + & \frac{1}{2}\sum_{i=1}^{n}T_i\tilde{f}^{*''}\left({\theta}_{true}^{\top} W^{1 / 2}\left(\Phi\left(X_{i}\right)-\bar{\Phi}\right)\right)(\left(\Phi\left(X_{i}\right)-\bar{\Phi}\right)^{\top}\Delta)^2  \\
        + & \frac{1}{6} \sum_{i=1}^{n}T_i\tilde{f}^{*'''}\left(\tilde{\theta}_{true}^{\top} W^{1 / 2}\left(\Phi\left(X_{i}\right)-\bar{\Phi}\right)\right) \left(\left(\Phi\left(X_{i}\right)-\bar{\Phi}\right)^{\top}\Delta\right)^3\\
        - & \delta^{\frac{1}{2}}||\Delta||_2 \quad ( \text{We treat $\left(\Phi\left(X_{i}\right)-\bar{\Phi}\right)^{\top}\Delta$ as a univariate variable and apply mean value theorem.})\\
        \geq &\frac{1}{2}\sum_{i=1}^{n}T_i\tilde{f}^{*''}\left({\theta}_{true}^{\top} W^{1 / 2}\left(\Phi\left(X_{i}\right)-\bar{\Phi}\right)\right)(\left(\Phi\left(X_{i}\right)-\bar{\Phi}\right)^{\top}\Delta)^2  \\
        = & \frac{1}{2}\sum_{i=1}^{n}T_i\tilde{f}^{*''}\left({\theta}_{true}^{\top} W^{1 / 2}\left(\Phi\left(X_{i}\right)-\bar{\Phi}\right)\right)(\left(\Phi\left(X_{i}\right)-\bar{\Phi}\right)^{\top}\Delta)^2  \\
        + & \frac{1}{6} \sum_{i=1}^{n}T_i\tilde{f}^{*''}\left(\tilde{\theta}_{true}^{\top} W^{1 / 2}\left(\Phi\left(X_{i}\right)-\bar{\Phi}\right)\right) \left(\left(\Phi\left(X_{i}\right)-\bar{\Phi}\right)^{\top}\Delta \right)^3 \\
        -&  \left(\delta^{\frac{1}{2}} + ||\sum_{i=1}^{n}T_i\tilde{f}^{*'}\left(\theta_{true}^{\top} W^{1 / 2}\left(\Phi\left(X_{i}\right)-\bar{\Phi}\right)\right)\left(\Phi\left(X_{i}\right)-\bar{\Phi}\right)||_2\right)||\Delta||_2. \quad (\text{Cauchy inequality})
        \end{aligned}
        $$
        Recall that the Bernstein's inequality for random matrices in \cite{tropp2015introduction} says the following. Let $\left\{M_k\right\}$ be a sequence of independent random matrices with dimensions $d_1 \times d_2$. Assume that $E (M_k)=0$ and $\left\|M_k\right\|_2 \leq R_n$ almost surely. Define
        $$
        \sigma_n^2=\max \left\{\left\|\sum_{k=1}^n E\left(M_k M_k^{\top}\right)\right\|_2,\left\|\sum_{k=1}^n E\left(M_k^{\top} M_k\right)\right\|_2\right\}
        $$
        Then for all $t \geq 0$,
        $$
        \operatorname{pr}\left(\left\|\sum_{k=1}^n M_k\right\|_2 \geq t\right) \leq\left(d_1+d_2\right) \exp \left(-\frac{t^2 / 2}{\sigma_n^2+R_n t / 3}\right) .
        $$
        Let $M_i = \frac{T_i}{\pi(X_i)}\left(\Phi\left(X_{i}\right)-\bar{\Phi}\right)$.   We derive that
        $$
        \begin{aligned}
        & E\left(M_i\right) = E\left(\frac{T_i}{\pi(X_i)}\left(\Phi\left(X_{i}\right)-\bar{\Phi}\right)\right) = E\left(E\left(\frac{T_i}{\pi(X_i)}\left(\Phi\left(X_{i}\right)-\bar{\Phi}\right) \mid X_i\right)\right) = 0, \\
        & ||M_i||_2 = ||\frac{T_i}{\pi(X_i)}\left(\Phi\left(X_{i}\right)-\bar{\Phi}\right)||_2 \leq ||\frac{T_i}{\pi(X_i)}||_2 ||\left(\Phi\left(X_{i}\right)-\bar{\Phi}\right)||_2 \leq \frac{C_2}{\eta^{1/2}} n^{1/2}, \\
        & E\left(||\sum_{i=1}^nM_iM_i^{\top}||_2\right) = E\left(||\sum_{i=1}^n\frac{T_i}{\pi(X_i)^2} \left(\Phi\left(X_{i}\right)-\bar{\Phi}\right) \left(\Phi\left(X_{i}\right)-\bar{\Phi}\right)^{\top}||_2\right) \\
        & \leq \sup_{x \in \mathcal{X}} \frac{T_i}{\pi(X_i)^2} E\left( ||\sum_{i=1}^n\left(\Phi\left(X_{i}\right)-\bar{\Phi}\right) \left(\Phi\left(X_{i}\right)-\bar{\Phi}\right)^{\top}||_2\right) =  \frac{n}{\eta^2} E\left(||\left(\Phi\left(X_{i}\right)-\bar{\Phi}\right) \left(\Phi\left(X_{i}\right)-\bar{\Phi}\right)^{\top}||_2\right) \\
        & \leq \frac{C_2n}{\eta^2}, \\
        & E\left(||\sum_{i=1}^nM_i^{\top}M_i||_2\right) = E\left(||\sum_{i=1}^n\frac{T_i}{\pi(X_i)^2} \left(\Phi\left(X_{i}\right)-\bar{\Phi}\right)^{\top} \left(\Phi\left(X_{i}\right)-\bar{\Phi}\right)||_2\right) \\
        & \leq \frac{n}{\eta^2} E\left(||\left(\Phi\left(X_{i}\right)-\bar{\Phi}\right)^{\top} \left(\Phi\left(X_{i}\right)-\bar{\Phi}\right)||_2\right) \leq \frac{C_2 n^{3/2}}{\eta^2}.
        \end{aligned}
        $$
        Therefore,
        $$
        \sigma_n^2=\max \left\{\left\|\sum_{i=1}^n E\left(M_i M_i^{\top}\right)\right\|_2,\left\|\sum_{i=1}^n E\left(M_i^{\top} M_i\right)\right\|_2\right\} = \frac{C_2 n^{3/2}}{\eta^2},
        $$
        and
        $$
        \operatorname{pr}\left(\left\|\sum_{i=1}^n M_i\right\|_2 \geq t\right) \leq\left(1+ K \right) \exp \left(-\frac{t^2 / 2}{\frac{C_2 n^{3/2}}{\eta^2}+C_2n^{1/2} t / 3}\right) .
        $$
  Letting $t = O_p(n^{ 3/4 + \epsilon})$ for arbitrary $\epsilon > 0$, we have
        $$
        \begin{aligned}
        & ||\sum_{i=1}^{n}T_i\tilde{f}^{*'}\left(\theta_{true}^{\top} W^{1 / 2}\left(\Phi\left(X_{i}\right)-\bar{\Phi}\right)\right)\left(\Phi\left(X_{i}\right)-\bar{\Phi}\right)||_2 \\
        = & ||\sum_{i=1}^{n}\frac{T_i}{\pi(X_i)}\left(\Phi\left(X_{i}\right)-\bar{\Phi}\right)||_2 \leq C_3 n.
        \end{aligned}
        $$

        We have
        $$
        \begin{aligned}
         &  \sum_{i=1}^{n}T_i\tilde{f}^{*''}\left({\theta}_{true}^{\top} W^{1 / 2}\left(\Phi\left(X_{i}\right)-\bar{\Phi}\right)\right)(\left(\Phi\left(X_{i}\right)-\bar{\Phi}\right)^{\top}\Delta)^2 \\
         \geq & \frac{1}{\eta} \sum_{i=1}^{n}(\frac{T_i}{\pi(X_i)}\left(\Phi\left(X_{i}\right)-\bar{\Phi}\right)^{\top}\Delta)^2\\
         = & \frac{1}{\eta} n \Delta^{\top} \frac{1}{n} \sum_{i=1}^{n}(\frac{T_i}{\pi(X_i)}\left(\Phi\left(X_{i}\right)-\bar{\Phi}\right))(\frac{T_i}{\pi(X_i)}\left(\Phi\left(X_{i}\right)-\bar{\Phi}\right))^{\top} \Delta.
        \end{aligned}
        $$

      According to Section 3.3.1 in \cite{lei2018asymptotics}, the smallest eigenvalue $\lambda_{1}$ of matrix $\frac{1}{n} \sum_{i=1}^{n} M_i M_i^\top$ satisfies  that $\lambda_{1} = O_p(\frac{1}{\log(n)^{\gamma}})$ for some $\gamma > 0$. Then,
        \begin{equation}
        n \Delta^{\top} \frac{1}{n} \sum_{i=1}^{n}(\frac{T_i}{\pi(X_i)}\left(\Phi\left(X_{i}\right)-\bar{\Phi}\right))(\frac{T_i}{\pi(X_i)}\left(\Phi\left(X_{i}\right)-\bar{\Phi}\right))^{\top} \Delta \geq n\lambda_{1} ||\Delta||^{2}_2 = O_p(\frac{n}{\log(n)^{\gamma}}) ||\Delta||^{2}_2.
        \label{eq:2}
        \end{equation}

        Finally, we have
        \begin{equation}
        \begin{aligned}
        & \frac{1}{6} \sum_{i=1}^{n}T_i\tilde{f}^{*'''}\left(\tilde{\theta}_{true}^{\top} W^{1 / 2}\left(\Phi\left(X_{i}\right)-\bar{\Phi}\right)\right) \left(\Delta^{\top}\left(\Phi\left(X_{i}\right)-\bar{\Phi}\right) \right)^3 \\
        \geq & - \frac{1}{6} \sum_{i=1}^{n}T_i\tilde{f}^{*'''}\left(\tilde{\theta}_{true}^{\top} W^{1 / 2}\left(\Phi\left(X_{i}\right)-\bar{\Phi}\right)\right) ||\Delta||^{3}_2 ||\left(\Phi\left(X_{i}\right)-\bar{\Phi}\right)||^{3}_2  \\
        \geq & - \frac{1-\eta}{6} ||\Delta||^{3}_2 \sum_{i=1}^{n} ||\left(\Phi\left(X_{i}\right)-\bar{\Phi}\right)||^{3}_2\\
        =& - \frac{1-\eta}{6} ||\Delta||^{3}_2 O_p(K^{3/2} n^{-3/2}).   
        \end{aligned}
        \label{eq:3}
        \end{equation}
    
      Using \eqref{eq:2} and \eqref{eq:3}, we conclude that
        $$
        \begin{aligned}
        & G\left(\theta_{true}+\Delta\right)-G\left(\theta_{true}\right) \\
        \geq & \frac{1}{\eta}n\lambda_{1} ||\Delta||^{2}_2  - (\delta^{\frac{1}{2}} + O_p(n^{ 3/4 + \epsilon}))||\Delta||_2  - \frac{1-\eta}{6}
        ||\Delta||^{3}_2 O_p(K^{3/2} n^{-3/2}) \\
        = & O_p(\frac{n}{\log(n)^{\gamma}})||\Delta||^{2}_2 - O_p(n^{3/4 + \epsilon})||\Delta||_2 - ||\Delta||^{3}_2 \\
        = & O_p(\frac{n^{\frac{1}{2} + 2\epsilon}}{\log(n)^{\gamma}}) - O_p(n^{1/2}) - O_p(n^{-3/4 + 3\epsilon}) \\
        \geq & O_p(n^{ \frac{1}{2} + 1.5\epsilon}) - O_p(n^{1/2}) - O_p(n^{-3/4 + 3\epsilon}). \quad (\text{ since $O_p(n^{0.5\epsilon} / \log(n)^{\gamma} \rightarrow \infty$ for arbitrary $\epsilon > 0$})
        \end{aligned}
        $$
For $\epsilon>0$ such that $\frac{1}{2} + 1.5\epsilon > -3/4 + 3\epsilon$, Lemma 5 is proved.

        Hence, we conclude that
        $$
        \begin{aligned}
        & \sup_{x \in \mathcal{X}}\left|\tilde{f}^{*'}\left(\hat{\theta}^{\top} W^{1 / 2}\left(\Phi\left(X_{i}\right)-\bar{\Phi}\right)\right)-\frac{1}{\pi(X_i)}\right| \\
        = & \sup_{x \in \mathcal{X}}\left|\tilde{f}^{*'}\left(\hat{\theta}^{\top} W^{1 / 2}\left(\Phi\left(X_{i}\right)-\bar{\Phi}\right)\right) - \tilde{f}^{*'}(m^{*}(X_i))\right| \\
        \leq & \sup_{x \in \mathcal{X}}\left|\tilde{f}^{*'}\left(\hat{\theta}^{\top} W^{1 / 2}\left(\Phi\left(X_{i}\right)-\bar{\Phi}\right)\right)-\tilde{f}^{*'}\left(\theta_{true}^{\top} W^{1 / 2}\left(\Phi\left(X_{i}\right)-\bar{\Phi}\right)\right)\right| \\
        + & \sup _{x \in \mathcal{X}}\left|\tilde{f}^{*'}\left(\theta_{true}^{\top} W^{1 / 2}\left(\Phi\left(X_{i}\right)-\bar{\Phi}\right)\right)-\tilde{f}^{*'}(m^{*}(X_i))\right| \\
        = & O\{\sup _{x \in \mathcal{X}}\left|(\theta_{true} - \hat{\theta})^{\top}W^{1 / 2}\left(\Phi\left(X_{i}\right)-\bar{\Phi}\right)\right\}+o_p(1) \\
        \leq & O\left\{\sup _{x \in \mathcal{X}}||(\theta_{true} - \hat{\theta})||_2||W^{1 / 2}\left(\Phi\left(X_{i}\right)-\bar{\Phi}\right)||_2\right\}+o_p(1) \\
        = & O_p\left\{n^{s-\frac{1}{2} + \epsilon}\right\}+o_p(1) \\
        = & o_p(1).
        \end{aligned}
        $$

        Therefore, we prove that  $\sup_{x \in \mathcal{X}}|\hat{w}_{i} - \frac{1}{\pi(X_i)}| = o_p(1)$.

		Following the empirical process arguments  by \cite{wang2020minimal} and \cite{fan2022optimal}, it holds that
		$
		A_{1}=o_p(n^{-1/{2}})
		$. Next,
		$$
		\begin{aligned}
			A_{2}
			&=\sum_{i=1}^{n}T_i\hat{w}^{MB}_{i}\left(\beta^{\top}W^{1 / 2}\Phi\left(X_{i}\right)-\frac{1}{n}\sum_{i=1}^{n}\beta^{\top}W^{1 / 2}\Phi\left(X_{i}\right)\right)\\
			&\leq \|\beta\|_2 \|\hat{w}^{MB}_{i}W^{1 / 2}\left(\Phi\left(X_{i}\right)-\bar{\Phi}\right)\|_2 \\
			&\leq O_p(n^{s+ \gamma + \alpha}) = o_p(n^{-1/2}).
		\end{aligned}
		$$
		The $R_1$ and $R_2$ are regular and asymptotically linear, and they determine the asymptotic expansion of $\hat{\tau}_1-\tau_1$. Similar expansion holds for $\hat{\tau}_0-\tau_0$. Using these asymptotic expansions, it follows  that the semiparametric efficiency bound for ATE estimation is attained.
	\end{proof}

	\section{Additional Information for   Numerical Studies}
	We provide more information about the numerical studies in the main article. In particular, we display the results of MB and MB2 by varying    $(\delta_1,\delta_0)$, where $\delta_1$ and $\delta_0$ are the threshold parameters for the treated and control groups, respectively. The outputs show that both MB and MB2 are not sensitive to the values of  the threshold parameters if they are small enough. The results are summarized in Table \ref{TableS1}.

	\begin{table}[t]
		
		\caption{Performance of MB and MB2 by varying  $(\delta_1,\delta_0)$ in Scenarios A, C, D}
		\label{TableS1}
		\resizebox{\textwidth}{!}{
			\begin{tabular}{cccccccccc}
				\toprule
				Scenario A & \multicolumn{1}{c}{{\textbf{Bias}}} & \multicolumn{1}{c}{\textbf{SD}} & \multicolumn{1}{c}{\textbf{RMSE}} & \multicolumn{1}{c}{\textbf{maxASMD}} & \multicolumn{1}{c}{\textbf{meanASMD}} & \multicolumn{1}{c}{\textbf{medASMD}}& \multicolumn{1}{c}{\textbf{GMIM}}  \\
				
				\midrule
				$\delta_1=\delta_0 = 1$ \\
				{\textbf{MB}}&-0.25&3.16&3.17&0.11&0.01&0.01&0.00\\
				{\textbf{MB2}}&-0.24& 3.16&3.17&0.11&0.01&0.01&0.00\\
				\midrule
				$\delta_1=\delta_0 = 10^{-2}$ \\			{\textbf{MB}}&-0.11&3.21&3.21&0.05&0.00&0.00&0.00\\
				{\textbf{MB2}}&-0.11&3.21&3.21&0.05&0.00&0.00&0.00\\ \hline
				$\delta_1=\delta_0 =  10^{-4}$ \\
				{\textbf{MB}}&-0.10&3.22&3.22&0.03&0.00&0.00&0.00\\
				{\textbf{MB2}}&-0.10& 3.22&3.22&0.03&0.00&0.00&0.00\\
				\midrule
				$\delta_1=\delta_0 =  10^{-6}$ \\			{\textbf{MB}}&-0.10&3.22&3.22&0.03&0.00&0.00&0.00\\
				{\textbf{MB2}}&-0.10& 3.22&3.22&0.03&0.00&0.00&0.00\\
				\midrule
				Scenario C &  \multicolumn{1}{c}{{\textbf{Bias}}} & \multicolumn{1}{c}{\textbf{SD}} & \multicolumn{1}{c}{\textbf{RMSE}} & \multicolumn{1}{c}{\textbf{maxASMD}} & \multicolumn{1}{c}{\textbf{meanASMD}} & \multicolumn{1}{c}{\textbf{medASMD}}& \multicolumn{1}{c}{\textbf{GMIM}}  \\

				\midrule
				$\delta_1=\delta_0 =  10^{-2}$
				\\
				{\textbf{MB}}&0.26&0.81&0.85&0.39&0.02&0.01&0.02\\
				{\textbf{MB2}}&0.37&0.86&0.93&0.30&0.03&0.01&0.03\\
				\midrule
				$\delta_1=\delta_0 = 10^{-4}$ \\
				{\textbf{MB}}&0.21&0.80&0.83&0.36&0.02&0.00&0.02\\
				{\textbf{MB2}}&0.26&0.85&0.89&0.27&0.02&0.00&0.02\\
				\midrule
				Scenario D	& \multicolumn{1}{c}{{\textbf{Bias}}} & \multicolumn{1}{c}{\textbf{SD}} & \multicolumn{1}{c}{\textbf{RMSE}} & \multicolumn{1}{c}{\textbf{maxASMD}} &\multicolumn{1}{c}{\textbf{meanASMD}}& \multicolumn{1}{c}{\textbf{medASMD}} &
				\multicolumn{1}{c}{\textbf{GMIM}} \\
				\midrule
				$\delta_1=\delta_0 = 10^{-2}$
				\\
				{\textbf{MB}}&-0.24&0.46 &0.52&0.15&0.01&0.00&0.01\\
				{\textbf{MB2}}&-0.21&0.45 &0.49&0.17&0.01&0.00&0.00\\
				\midrule
				$\delta_1=\delta_0 =  10^{-4}$
				\\
				{\textbf{MB}}&-0.16&0.46 &0.49&0.13&0.01&0.00&0.00\\
				{\textbf{MB2}}&-0.15&0.45 &0.47&0.17&0.01&0.00&0.00\\
				\bottomrule
			\end{tabular}
		}
	\end{table}
	
		\section{Alternative formulations of Mahalanobis balancing}\label{sec3.6}
	
	We discuss the formulation of Mahalanobis balancing when the normalization constraint is added to Problem (5).   Using the entropy function as the loss function, the optimization problem is
	\begin{equation*}
		\begin{array}{ll}
			\underset{w \in \mathbb{R}^{n_1}}{\text{minimize}}:& \ \ \sum_{i=1}^{n}T_iw_i\log(w_i)\\
			\text{subject to:}&\left\{
			\begin{array}{lll}
				w_i\geq 0,\ \  i\in\{j: T_j=1\}; \\
				\sum_{i=1}^{n}T_i w_i=1;\\
				\sum_{i=1}^{n}T_i\{w_i\Phi(X_i)-\bar{\Phi}\}^{\top}{W}\sum_{i=1}^{n}T_i\{w_i\Phi(X_i)-\bar{\Phi}\} \leq \delta.
			\end{array}
			\right.
		\end{array}
	\end{equation*}
	By  the Fenchel duality theory, the corresponding dual problem is
	\begin{equation}
		\underset{\theta \in \mathbb{R}^K}{\text{minimize}}:
		\ -\log\left[\sum_{i=1}^{n}T_i\exp\left\{\theta^\top {W}^{1/2}\Phi(X_i)\right\}\right]+\theta^\top {W}^{1/2}\bar{\Phi}+\sqrt{\delta}\lVert \theta\rVert_2.\label{dualnormalized}
	\end{equation}
	
	The estimated balanced weight for subject $i$ in the treated group is
	$$\hat{w}_i^{MB}=\frac{\exp\left\{\hat{\theta}^\top {W}^{1/2}\Phi(X_i)\right\}}{\sum_{j:T_j=1}\exp\left\{\hat{\theta}^\top {W}^{1/2}\Phi(X_j)\right\}}, \ \ \text{for all}\ \ i\in\{j: T_j=1\},$$
	where $\hat{\theta}$ is solved from the dual problem. When $\delta=0$,  the dual problem  \eqref{dualnormalized}  reduces to that by entropy balancing \citep{hainmueller2012entropy,zhao2016entropy}. Therefore, Mahalanobis balancing with the normalization constraint is a direct generalization of  entropy balancing. The  parameter $\delta$ can be more difficult to tune when the normalization constraint is inserted. 
	
	\bibliographystyle{ecta}
	\small{
		\bibliography{paper-ref}
	}